\documentclass[proc]{edpsmath}
\usepackage{float}
\usepackage{amssymb,amsbsy,amsmath,amsfonts,amscd}
\usepackage[utf8]{inputenc}
\usepackage[T1]{fontenc}
\usepackage{fancyhdr}
\usepackage{graphicx}
\usepackage[usenames, dvipsnames]{xcolor}

\def\R{{\mathbb R}}
\def\N{{\mathbb N}}
\newcommand{\bean}{\begin{eqnarray*}} 
	\newcommand{\eean}{\end{eqnarray*}} 
\newcommand{\bea}{\begin{eqnarray}}
\newcommand{\eea}{\end{eqnarray}} 

\newcommand{\valabs}[1]{{\left|#1\right|}}

\def\ds{\displaystyle}

\newtheorem{example}{Example}
\newtheorem{prop}{Proposition}
\newtheorem{rem}{Remark}
\newcommand{\gl}{\lambda}  
\newcommand{\gge}{\varepsilon}
\newcommand{\nsub}[1]{\vert\vert\vert #1 \vert\vert\vert}

\usepackage{xargs}                      
\usepackage[]{todonotes}
\newcommand{\modR}[1]{#1}
\newcommand{\modRR}[1]{\textcolor{black}{#1}}
\newcommand{\modC}[1]{\textcolor{black}{#1}}
\newcommand{\modF}[1]{\textcolor{black}{#1}}
\newcommand{\modFF}[1]{\textcolor{black}{#1}}

\newcommand{\modM}[1]{#1}
\newcommand{\modMM}[1]{\textcolor{black}{#1}}

\begin{document}

\title{Multiscale population dynamics in reproductive biology: singular perturbation reduction in deterministic and stochastic models}
\author{Celine Bonnet}\address{CMAP, Ecole Polytechnique, Route de Saclay, 91128 Palaiseau cedex, France. E-mail adress: celine.bonnet@polytechnique.edu}
\author{Keltoum Chahour}\address{LERMA, Mohammadia Engineering School, Mohamed V University in Rabat, Morocco\\LJAD Université Côte d'Azur, Parc Valrose, 06108 Nice, France}
\author{Frédérique Clément}\address{Inria, Université Paris-Saclay\\LMS, Ecole Polytechnique, CNRS, Université Paris-Saclay}
\author{Marie Postel}\address{Sorbonne Universit\'e, Universit\'e Paris-Diderot SPC, CNRS, Laboratoire Jacques-Louis Lions, LJLL}
\author{Romain Yvinec}\address{PRC,
	INRA, CNRS, IFCE, Université de Tours, 37380 Nouzilly, France}
%
%
\begin{abstract}
\modF{In this study, we describe different modeling approaches for ovarian follicle population dynamics, based on either ordinary (ODE), partial (PDE) or stochastic (SDE) differential equations, and accounting for interactions between follicles. We put a special focus on representing the population-level feedback exerted by growing ovarian follicles onto the activation of quiescent follicles. We take advantage of the timescale difference existing between the growth and activation processes to apply model reduction techniques in the framework of singular perturbations. We first study the linear versions of the models to derive theoretical results on the convergence to the limit models. In the nonlinear cases, we provide detailed numerical evidence of convergence to the limit behavior. We reproduce the main semi-quantitative features characterizing the ovarian follicle pool, namely a bimodal distribution of the whole population, and a slope break in the decay of the quiescent pool with aging.}
\end{abstract}
\begin{resume}
\modF{Dans cette étude, nous décrivons différentes approches de modélisation de la dynamique des populations de follicules ovariens, basées sur des équations différentielles ordinaires (EDO), aux dérivées partielles (EDP) ou stochastiques (SDE), et tenant compte des interactions entre follicules. Nous avons mis un accent particulier sur la représentation des rétro-actions exercées par les follicules en croissance sur l'activation des follicules quiescents. Nous tirons parti de la différence d'échelle de temps entre les processus de croissance et d’activation pour appliquer des techniques de réduction de modèle dans le cadre des perturbations singulières. Nous étudions d’abord les versions linéaires du modèle afin d’en déduire des résultats théoriques sur la convergence vers le modèle limite. Dans le cas non linéaire, nous fournissons des arguments numériques détaillés sur la convergence vers le comportement limite. Nous reproduisons les principales caractéristiques semi-quantitatives caractérisant le pool de follicules ovariens, à savoir une distribution bimodale  de la population totale et une rupture de pente dans la décroissance du pool de follicules quiescents avec le vieillissement.}
\end{resume}
\maketitle
\section*{Introduction}

In mammals, the pool of oocytes (egg cells) available for a female throughout her reproductive life is fixed very early, either during the fetal life or in the perinatal period.  All along their maturation, oocytes are sheltered within spheroidal somatic structures called ovarian follicles. Folliculogenesis is the process of growth and maturation undergone by ovarian follicles from the time they leave the pool of quiescent, primordial follicles until ovulation, when they release a fertilizable oocyte. 

Follicle growth is first due to the enlargement of the oocyte, then to the proliferation of somatic cells organized into successive concentric cell layers, and finally to the inflation of a fluid-filled cavity (antrum) that forms above a critical size. The activation of primordial follicles can occur at any time once they are formed \cite{monniaux_18b}\modF{, even if they can remain quiescent for up to tens of years \cite{reddy_10}}. Growing follicles can progress along the first developmental stages (known as ``basal development'') before puberty. The final developmental stages (known as ``terminal development'') can only occur after puberty ; they are related to the dynamics of ovarian cycles, involving endocrine feedback loops between the ovaries on one side, and the hypothalamus and pituitary gland on the other side. The whole sequence of development spans several months, as assessed by cell kinetics studies \cite{gougeon_96} or grafts  of ovarian cortex \cite{donnez_14}. The terminal stages are the shortest ; they cover a few weeks at most.

Since follicle activation is asynchronous, all developmental stages can be observed in the ovaries at any time during reproductive life. The follicle distribution (mostly studied using the size as a maturity marker) has a characteristic bimodal pattern, which is remarkably preserved between species. This pattern remains similar with ovarian aging, yet with a decreased amplitude  \cite{livre_bio_INRA_folliculogenese}, as a result of the progressive exhaustion of the quiescent pool. Such a distribution is shaped not only by the differences in the follicle activation times, but also by the hormonal interactions between follicles \cite{monniaux_14}. In particular, the activation and growth rates in the earliest stages are moderated by the Anti-M\"ullerian Hormone (AMH) secreted locally by the subpopulation of ``intermediary'' follicles (rigorously speaking: from the fully activated one-layer stage to the pre-antral and small antral stages) \cite{visser_05}. At the other end, the selection of ovulatory follicles results from a competition-like process operating amongst terminally developing follicles \cite{clement_13b}, which is mediated by endocrine controls and associated with a species-specific number of ovulations. Namely, inhibin (a peptid hormone) and estradiol, produced by the mature follicles, feedback onto the pituitary gland, leading to a drop in a pituitary hormone (the Follicle-Stimulating Hormone) supporting follicle survival.

Less than one in a thousand of the follicles manage to reach the ovulatory stage. All others \modF{disappear} through a degeneration process (atresia) associated with the death of the somatic cells (during mid and terminal folliculogenesis) or oocyte (in the quiescent pool and during early folliculogenesis). For instance, in humans, the quiescent pool size is of the order of 1 million follicles, amongst which only some hundreds will reach ovulation \cite{monniaux_18}.\\

\modF{Experimentalist} investigators have proposed classifying follicle development into different stages, according to morphological and functional criteria such as follicle and oocyte diameters, number of cell layers, number of somatic cells, antrum formation \cite{gougeon_96,pedersen_68}. Hence, a natural formalism to consider when modeling follicle population dynamics is that of compartmental modeling, using either deterministic or stochastic rates for transfer ($\lambda_i$) and exit ($\mu_i$) rates (see Equation \eqref{fig:compartment}).
Pioneering studies (see e.g. \cite{faddy_76}) have focused on fitting the parameters entering these rates according to follicle numbers available in each developmental stage. However, these studies remained rather descriptive and considered at best possibly time-varying (piecewise constant) rates \cite{faddy_88}, yet with no follicle interaction.

\begin{equation}\label{fig:compartment}
\begin{CD}
	F_0 	@>\lambda_0>> F_1 	@>\lambda_1>> \cdots  	@>\lambda_{N-1}>> F_{N} \\
	@VV\mu_0V  @VV\mu_1V @VV\cdots V @VV\mu_NV\\
	\emptyset @.  \emptyset @.  \emptyset @.  \emptyset
\end{CD}
\end{equation}

Most of the classification criteria change in a continuous manner. In addition, the most common variable available to monitor follicle development on the ovarian scale, follicle size, is an intrinsically continuous variable. Hence, another suitable modeling formalism is that of PDE models for structured population dynamics. Although the interest of such a formalism has been pointed out quite early \cite{mariana_77}, it has yet not been implemented. 

Finally, in some situations, a discrete stochastic formalism can be useful both to handle finite-size effects and follow individual follicle trajectories. This is especially true for the relatively small cohort of terminally developing follicles, and for transient physiological regimes when follicle pools are either still replenishing, or, on the contrary getting progressively exhausted. In any case, such a formalism gives insight into the fluctuations around the average deterministic behavior. 

In this study, we describe different modeling approaches for follicle population dynamics, based on either ODE, PDE or SDE, and accounting for interactions between follicles. We put a special focus on representing the population-level modulation exerted by growing ovarian follicles on the activation of quiescent, primordial follicles. We take advantage of the timescale difference between the growth and activation processes to apply model reduction techniques in the framework of singular perturbations (slow/fast systems). 

The paper is organized as follows. We successively introduce the ODE, PDE and SDE formulation of the model for follicle population dynamics. We describe the initial (non-rescaled) model in the ODE case. In each case, we introduce (i) the model in rescaled timescale exhibiting a slow/fast structure with a small perturbation parameter ($\varepsilon > 0$) and (ii) the model in the limit $\varepsilon\to 0$. We discuss the well-posedness of the limit models in two situations: the linear formulation and a weakly nonlinear formulation in which only the quiescent follicle population is subject to a feedback from the remaining of the population. In the linear case, we prove the convergence of the rescaled to the limit models. In the nonlinear case, we provide detailed numerical evidence of convergence. The numerical illustrations are settled within a biologically-realistic framework, allowing us to reproduce the main semi-quantitative features characterizing the dynamics of the ovarian follicle pool, namely a bimodal distribution of the whole population and a slope break in the decay of the quiescent pool with aging.

	\section{\modF{Compartmental, ODE-based} model}\label{sec:ode}

	\subsection{Initial model}
	Starting from the schematic model (see \modM{Eq.} \eqref{fig:compartment}), we formalize a system of nonlinear ordinary differential equations (ODE) as follows. Let $d\in \N^*$ and $y=(y_0,\dots, y_d)$ be a function such that, for all $i\in \{0,\dots,d\}$, $y_i : t\in \mathbb{R}_+ \mapsto y_i(t) $ represents the time evolution of the number of follicles of maturity $i$. Follicles in the first compartment ($i=0$) are named quiescent follicles, and their maturation and death rates are denoted by $\bar{\gl_0}$ and $\bar{\mu_0}$, respectively. Follicles in the \modF{intermediate} compartments ($1\leq i \leq d-1$) are the growing follicles, and may either mature and go to the next maturation stage $i+1$, at rate $\gl_i$, or die at rate $\mu_i$. Follicles in the last compartment ($i=d$) are named the mature follicles and can only die at rate $\mu_d$, \textit{i.e.} $\gl_d=0$ 
	(death in this compartment corresponds to either degeneration or ovulation).
	All the rates $(\mu_i,\lambda_i)$ may depend on the growing and mature follicles population (non-local interactions), which leads
to the following nonlinear ODE system
	\begin{equation}\label{eq:ode_original}
	\left\lbrace \begin{array}{lll}
	\ds\frac{d y_0(t)}{dt} &\ds=& \ds-\big(\bar{\gl_0}(y(t)) + \bar{\mu_0}(y(t)) \big) y_0(t) \,,\\ \\
	\ds\frac{d y_1(t)}{dt} &\ds=&\ds \bar{\gl_0}(y(t)) y_{0}(t)-\big(\gl_1(y(t)) + \mu_1(y(t)) \big) y_1(t),\\ \\
	\ds\frac{d y_i(t)}{dt} &\ds=& \ds\gl_{i-1}(y(t)) y_{i-1}(t)-\big(\gl_i(y(t)) + \mu_i(y(t)) \big) y_i(t), \quad i \in \{2,\dots,d\}
	\end{array}\right.
	\end{equation}
	where, for $i=0$, 
	\begin{equation}\label{eq:lambda0_mu0_ode}
	\bar{\gl_0}(y)=\frac{\bar{f_0}}{1+K_{1,0}\ds\sum_{j=1}^d a_j y_j},\quad \bar{\mu_0}(y)=\bar{g_0}\left({1+K_{2,0}\ds\sum_{j=1}^d b_{j} y_j}\right)\,,
	\end{equation}
	with non-negative parameter constants $\bar{f_0}$, $\bar{g_0}$, $K_{1,0}$, $K_{2,0}$, and $a_j\in [0,1], b_j\in [0,1]$.
	
	For $i\in \{1,\dots,d\}$,
	\begin{equation}\label{eq:lambdai_mui_ode}
	\gl_{i}(y)=\dfrac{f_i}{1+K_{1,i}\ds\sum_{j=1}^d \omega_{1,j}y_j},\quad \mu_i(y)= g_i\left({1+K_{2,i}\ds\sum_{j=1}^d \omega_{2,j}y_j}\right)\,,
	\end{equation}
	with non-negative parameter constants, $f_i$ ($f_d=0$), $K_{1,i}$, $g_i$, $K_{2,i}$ and $\omega_{1,j} \in [0,1]$, $\omega_{2,j} \in [0,1]$. The specific functional forms of the rate coefficients are motivated by biological knowledge (see Introduction). The population feedback tends to lower the maturation rate and to raise the death rate.
	With a non-negative vector $y^{in}\in \R_+^{d+1}$ as initial data, one \modF{can see} that Eq.~\eqref{eq:ode_original} generates \modF{a unique} non-negative solution for all times (the right-hand side is globally Lipschitz, with positive off-diagonal entries). Moreover, one can obtain immediately the following conservation law,
	\begin{equation}\label{eq:ode_original_conservation}
	\frac{d }{dt}\sum_{i=0}^d y_i(t) = - \bar{\mu_0}(y(t)) y_0(t)-\sum_{i=1}^d \mu_i(y(t)) y_i(t)\leq 0\,,
	\end{equation}
	which shows that any follicle sub-population $y_i$ is globally bounded.
	
	\subsection{Rescaled model}
	
	As outlined in the Introduction,  before reproductive \modF{senescence}, quiescent follicles are \modF{very numerous} compared to the growing and mature follicles, \modF{follicle activation dynamics} are much slower \modF{than growth dynamics}, \modR{yet the flow of follicles between each compartment is of the same order}. \modF{In consistency with this timescale contrast, we} introduce a small positive parameter $\varepsilon \ll 1$, such that
	\begin{equation}
	\bar{f_0} = \gge f_0\,, \quad  \bar{g_0} = \gge g_0 \,, \quad y_0^{in}= \frac{x_0^{in}}{\gge}
	\end{equation} 
	with non-negative constants $f_0$, $g_0$ and positive initial data $x_0^{in}$, independent of $\gge$. Note that the initial flow $\bar{f_0}y_0^{in}=f_0x_0^{in}$ is preserved.
	
	We then define the rescaled solution $x=(x_0,\dots,x_d)$ by, for all $t \geq 0$,
	\begin{equation}
	x_0(t) = \gge y_0(t/\gge), \quad \text{ and for all } i\geq 1\,,\quad x_i(t) = y_i(t/\gge).
	\end{equation}
	Then $x$ is solution of the following system
	\begin{equation}\label{eq:ode_rescaled}
	\left\lbrace \begin{array}{lll}
	\ds\frac{d x_0(t)}{dt} &\ds=&\ds -\big(\gl_0(x(t)) + \mu_0(x(t)) \big) x_0(t) \\ \\
	\ds\gge \frac{d x_i(t)}{dt} &\ds=& \ds\gl_{i-1}(x(t)) x_{i-1}(t)-\big(\gl_i(x(t)) + \mu_i(x(t)) \big) x_i(t), \quad i \in \{1,\dots,d\}
	\end{array}\right.
	\end{equation}
	with initial condition given by $x_0(t=0)=x_0^{in}$, and  $x_i(t=0)=x_i^{in}:=y_i^{in}$, for $i \in \{1,\dots,d\}$, and where, for $i=0$, 
	\begin{equation}\label{eq:lambda0_mu0_ode_rescaled}
	\gl_0(x)=\frac{f_0}{1+K_{1,0}\ds\sum_{i=1}^d a_ix_i},\quad \mu_0(x)=g_0\left({1+K_{2,0}\ds\sum_{j=1}^d b_{j}x_j}\right)\,,
	\end{equation}
	and \modF{$\gl_{i}$ and $\mu_i$, for $i \in \{1,\dots,d\}$, are defined in Eq.~\eqref{eq:lambdai_mui_ode}}. We note that the conservation law \eqref{eq:ode_original_conservation} becomes
	\begin{equation}\label{eq:ode_rescaled_conservation}
	\frac{d }{dt} x_0(t) + \gge\frac{d }{dt}\sum_{i=1}^d x_i(t) = - \sum_{i=0}^d \mu_i(x(t)) x_i(t)\leq 0\,,
	\end{equation}
	
	We now consider the limit for which the small parameter $\gge$ tends to $0$ and \modF{the associated sequence} $(x^\gge)$ solution of system \eqref{eq:ode_rescaled}. In such \modF{a} case, \modF{system \eqref{eq:ode_rescaled}} is called a ``slow-fast'' system ($x_0^\gge$ is the slow variable, $(x_1^\gge,\dots,x_d^\gge)$ are the fast variables) and the \modF{study of the} limit behavior when $\gge\to 0$ is a singular perturbation problem \modR{(see for instance \cite{Wasow1988})}.
	
	\subsection{Limit model}
	
	Formally, setting $\gge=0$ in \modF{system \eqref{eq:ode_rescaled} leads to}  the following system:
	\begin{equation}\label{eq:ode_limit}
	\left\lbrace \begin{array}{lll}
	\ds\frac{d \bar {x}_0(t)}{dt} &\ds=&\ds -\big(\gl_0(\bar {x}(t)) + \mu_0(\bar {x}(t)) \big) \bar {x}_0(t) \\
	\ds 0 &\ds=& \ds\gl_{i-1}(\bar {x}(t)) \bar {x}_{i-1}(t)-\big(\gl_i(\bar {x}(t)) + \mu_i(\bar {x}(t)) \big) \bar {x}_i(t), \quad i \in \{1,\dots,d\}
	\end{array}\right.
	\end{equation}
	with initial  condition given by $\bar x_0(t=0)=x_0^{in}$, and  $\bar x_i(t=0)=\bar x_i^{in}\geq 0$, that satisfies the second line of Eq.~\eqref{eq:ode_limit} at $t=0$. Note that system \eqref{eq:ode_limit} is a differential-algebraic system, in which the variable $(\bar {x}_1,\dots,\bar {x}_d)$ can be seen as \modF{reaching} instantaneously (at any time $t$) a quasi-steady state, ``driven'' by the time-dependent variable $\bar {x}_0(t)$.
	
	System \eqref{eq:ode_limit} is not necessarily well-posed, as there may be several solutions \modF{to} the second line of Eq.~\eqref{eq:ode_limit}.  \modF{In the next two specific examples, we can prove that system \eqref{eq:ode_limit} does admit a single positive solution, which is a natural} \modF{limit candidate} for the sequence $x^\gge$.

	\begin{example}[Linear case]\label{ex:ode_limit_linear}
		Let us suppose that $K_{1,i}=K_{2,i}=0$ for all $i\in \{0,\dots,d\}$, \modR{and $f_i+g_i>0$ for all $i\in \{1,\dots,d\}$}. Then, system \eqref{eq:ode_limit} becomes linear, and has a unique solution given by \modC{$\bar x =(\bar x_0,\dots,\bar x_d)$ such that for all $t \geq 0$,}
		\begin{equation}\label{eq:ode_limit_linear}
		\left\lbrace \begin{array}{lll}
		\ds\bar {x}_0(t) &\ds=&\ds x_0^{in}\exp\left(-(f_0+g_0)t\right)\,,\\
		\ds \bar  x_i(t) &\ds=& \ds \left(\prod_{j=0}^{i-1}\frac{f_j}{f_{j+1}+g_{j+1}}\right)\bar {x}_0(t), \quad i \in \{1,\dots,d\}\,.
		\end{array}\right.
		\end{equation}
	\end{example}
	
	\begin{example}[Feedback onto quiescent follicle activation and death rates]
		Let us suppose that \modF{$K_{1,i}=0$, $K_{2,i}=0$} \modR{and $f_i+g_i>0$} for $i\in \{1,\dots,d\}$, yet $K_{1,0}>0$ and $K_{2,0}\geq 0$. Then, system \eqref{eq:ode_limit} \modR{with positivity requirement ($x_i\geq 0$ for $i\in \{0,\dots,d\}$)} can be rewritten as:
		\begin{equation}\label{eq:ode_limit_K_1_0}
		    \left\lbrace \begin{array}{lll}
		    \ds\frac{d \bar {x}_0(t)}{dt} &\ds=&\ds -\big(\gl_0(\bar {x}(t)) + \mu_0(\bar {x}(t)) \big) \bar {x}_0(t)\,, \\
		\ds \bar  x_1(t) &\ds=& \ds \frac{-(f_1+g_1)+\sqrt{(f_1+g_1)^2 +4f_0\bar {x}_0(t)(f_1+g_1)K_{1,0}\sum_{i=1}^d \,a_i \prod_{j=1}^{i-1}\frac{f_j}{f_{j+1}+g_{j+1}}}}{2(f_1+g_1)K_{1,0}\sum_{i=1}^d \, a_i\prod_{j=1}^{i-1}\frac{f_j}{f_{j+1}+g_{j+1}}}\,,\\
		\ds \bar  x_i(t) &\ds=& \ds \left(\prod_{j=1}^{i-1}\frac{f_j}{f_{j+1}+g_{j+1}}\right)\bar {x}_1(t), \quad i \in \{2,\dots,d\}\,.
		\end{array}\right. 
		\end{equation}
		\modF{which admits a unique solution. }
		Indeed, one can verify that $\bar  x_1$ is the \modR{only} positive root of a polynomial of degree $2$, namely
		\[f_0\bar  x_0=(f_1+g_1)\left(1+K_{1,0}\sum_{i=1}^d \, a_i\prod_{j=1}^{i-1}\frac{f_j}{f_{j+1}+g_{j+1}}\bar  x_1\right)\bar  x_1\]
		
	\end{example}

	\subsection{Convergence in the linear case}
	
	In this paragraph, we assume that $K_{1,i}=K_{2,i}=0$ for all $i\in \{0,\dots,d\}$, \modR{and $f_i+g_i>0$ for all $i\in \{1,\dots,d\}$} as in Example \ref{ex:ode_limit_linear}. In such a case, one can solve system \eqref{eq:ode_rescaled} explicitly for each $\gge>0$. The solution is given by, using vectorial notations, for all $t \geq 0$,
	\begin{equation}\label{eq:ode_rescaled_solution_linear}
	\left\lbrace \begin{array}{lll}
	\ds x_0^{\gge}(t) &\ds=&\ds x_0^{in}\exp\left(-(f_0+g_0)t\right), \\
	\ds (x_1^{\gge},\cdots,x_d^{\gge})^T(t) &\ds=& \ds  e^{-\frac{B\,t}{\gge}}(x_1^{in},\cdots,x_d^{in})^T + \int_0^t \frac{1}{\gge}e^{-\frac{B(t-s)}{\gge}}e_1 f_0 x_0^{\gge}(s)ds, 
	\end{array}\right.
	\end{equation}
	where $e_1=(1,0,\cdots,0)^T$ and 
	$$ B=\begin{pmatrix}
	(f_1+g_1) & 0 & \dots & 0 \\
	-f_1 & \ddots & & \vdots\\
	0 & \ddots & \ddots & \vdots\\
	\vdots & &  & 0\\
	0 & 0 & -f_{d-1} &(f_d+g_d)\\
	\end{pmatrix}.$$
	
	It is thus clear that $x_0^{\gge}=\bar {x}_0$ is a constant sequence in $\varepsilon$ (as both $x_0^{\gge}$ and $\bar {x}_0$ have same initial conditions, and same evolution equation). For the fast variables, we prove the following
	\begin{prop}\label{prop:ode_cv_linear}Assume that $K_{1,i}=K_{2,i}=0$ for all $i\in \{0,\dots,d\}$, and $f_i+g_i>0$ for $i\in \{1,\dots,d\}$. Then, for all $\eta>0$, we have
		\begin{equation}
		\lim_{\gge \to 0} \sup_{t>\eta} \max_{i\in \{1,\dots,d\}}  \mid x_i^{\gge}(t) -  \bar  x_i(t) \mid =0\,.
		\end{equation}	
		where $(\bar  x_1,\cdots,\bar  x_d)$ is given in Eq.~\eqref{eq:ode_limit_linear}.
	\end{prop}
	\begin{proof}
		\modF{From} Eq.~\eqref{eq:ode_rescaled_solution_linear} and initial condition, \modF{and} using integration by parts, we obtain
		\begin{equation*}
		\int_0^t \frac{1}{\gge}e^{-\frac{B(t-s)}{\gge}}e_1 f_0 x_0^{\gge}(s)ds =  \int_0^t B^{-1}e^{-\frac{B(t-s)}{\gge}}e_1 f_0(f_0+g_0) x_0^{\gge}(s)ds + B^{-1}e_1 f_0 x_0^{\gge}(t)-B^{-1}e^{-\frac{Bt}{\gge}}e_1 f_0 x_0^{in}\,.
		\end{equation*}
		As $x_0^\gge$ is uniformly bounded (both in $\gge$ and time) by $x_0^{in}$, we obtain, taking the $1$-norm,
		\begin{multline*}
		\left\lVert \int_0^t B^{-1}e^{-\frac{B(t-s)}{\gge}}e_1 f_0(f_0+g_0) x_0^{\gge}(s)ds \right\rVert \leq   f_0(f_0+g_0)x_0^{in} \left\lVert \int_0^t B^{-1}e^{-\frac{B(t-s)}{\gge}}e_1ds\right\rVert \\
		\leq  f_0(f_0+g_0)x_0^{in}\gge \left\lVert B^{-2}(Id-e^{-\frac{Bt}{\gge}})e_1  \right\rVert \leq   f_0(f_0+g_0)x_0^{in}\gge \nsub{B^{-2}} (1+e^{-\min(f_i+g_i)\frac{t}{\gge}})
		\end{multline*}
		with $\nsub{B^{-2}} = \displaystyle \sup_{x \in \mathbb{R}^d, x\neq 0} \frac{\left\lVert B^{-2} x\right\rVert}{\left\lVert x\right\rVert}$. \modF{The third inequality above} was deduced from $$ \forall t \geq 0, \quad \left\lVert e^{-\frac{Bt}{\gge}}e_1  \right\rVert \leq \left\lVert e_1  \right\rVert e^{-\min(f_i+g_i)\frac{t}{\gge}}
		$$
		We verify that \modC{for all $t \geq 0$}, $B^{-1}e_1 f_0 x_0^{\gge}(t) =(\bar  x_1,\cdots,\bar  x_d)(t)$. Then, we obtain, 
		\begin{multline}
		\sup_{t >\eta } \lVert ( x_1^\gge,\cdots, x_d^\gge)(t)- (\bar  x_1,\cdots,\bar  x_d)(t)\rVert \leq \left(\lVert (x_1^{in},\cdots,x_d^{in})\rVert+ \nsub{B^{-1}} f_0 x_0^{in} \right) e^{-\min(f_i+g_i)\frac{\eta }{\gge}}\\+2f_0(f_0+g_0)x_0^{in} \nsub{ B^{-2} }  \gge,
		\end{multline}
		which converges to $0$ as $\gge$ converges to $0$.
	\end{proof}	

	\begin{rem}\label{rem:ode}
		The proof of Proposition \ref{prop:ode_cv_linear} can also be obtained as a direct application of Tikhonov theorem \cite{Wasow1988}. 
	\end{rem}
	
	\begin{rem}\label{rem:ode2}
		It is apparent in formula \eqref{eq:ode_rescaled_solution_linear} that one cannot hope to obtain a convergence on a time interval starting from $0$ (unless the initial data is ``well-prepared''), and that standard Ascoli-Arzela theorem would not apply in such a case, as the time derivative of $( x_1^\gge,\cdots, x_d^\gge)(t)$ is not uniformly bounded as $\gge\to 0$. 
	\end{rem}

\subsection{Numerical convergence}
In this paragraph, we illustrate the convergence of $(x_0^\gge, x_1^\gge,\cdots, x_d^\gge)$ towards $(\bar  x_0,\bar  x_1,\cdots,\bar  x_d)(t)$ in a nonlinear scenario. The chosen scenario and the parameter values are detailed in \modF{the Appendix} (section \ref{sec:annex}).

\begin{figure}
	\includegraphics{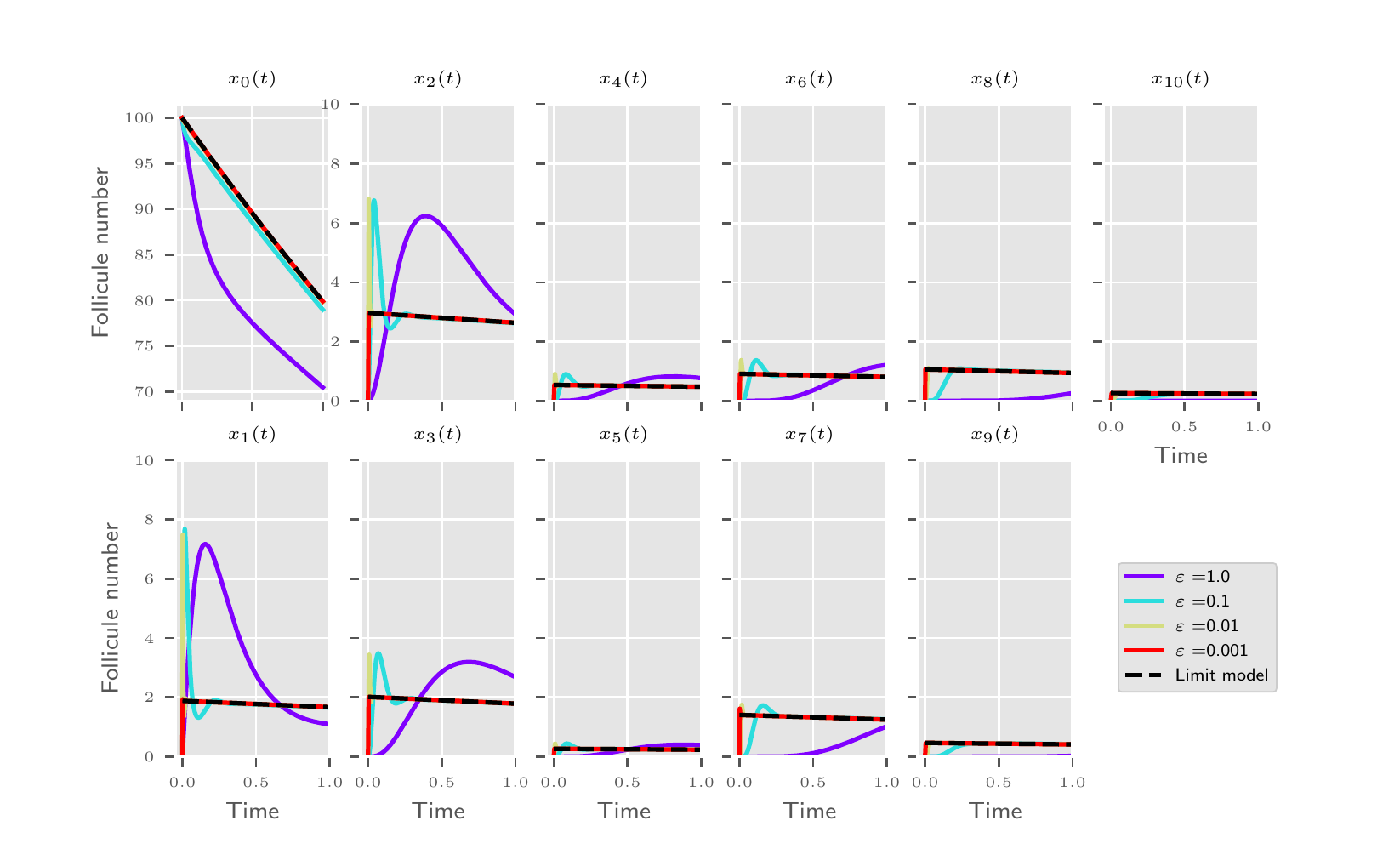}
	\caption{Trajectories \modF{in each maturity compartment} ($d=10$), for the rescaled variables $x_i^\varepsilon$, for different $\varepsilon$ (\modF{solid} colored lines, see legend \modF{insert}) and the reduced limit \modF{variables} $\bar x_i$ (\modF{black} dashed lines).\modR{ See the Appendix (section \ref{sec:annex}) for details on the parameter values used in the numerical simulations.}}
	\label{fig:ode_lim_mod}
\end{figure}
	
\begin{figure}
	\includegraphics{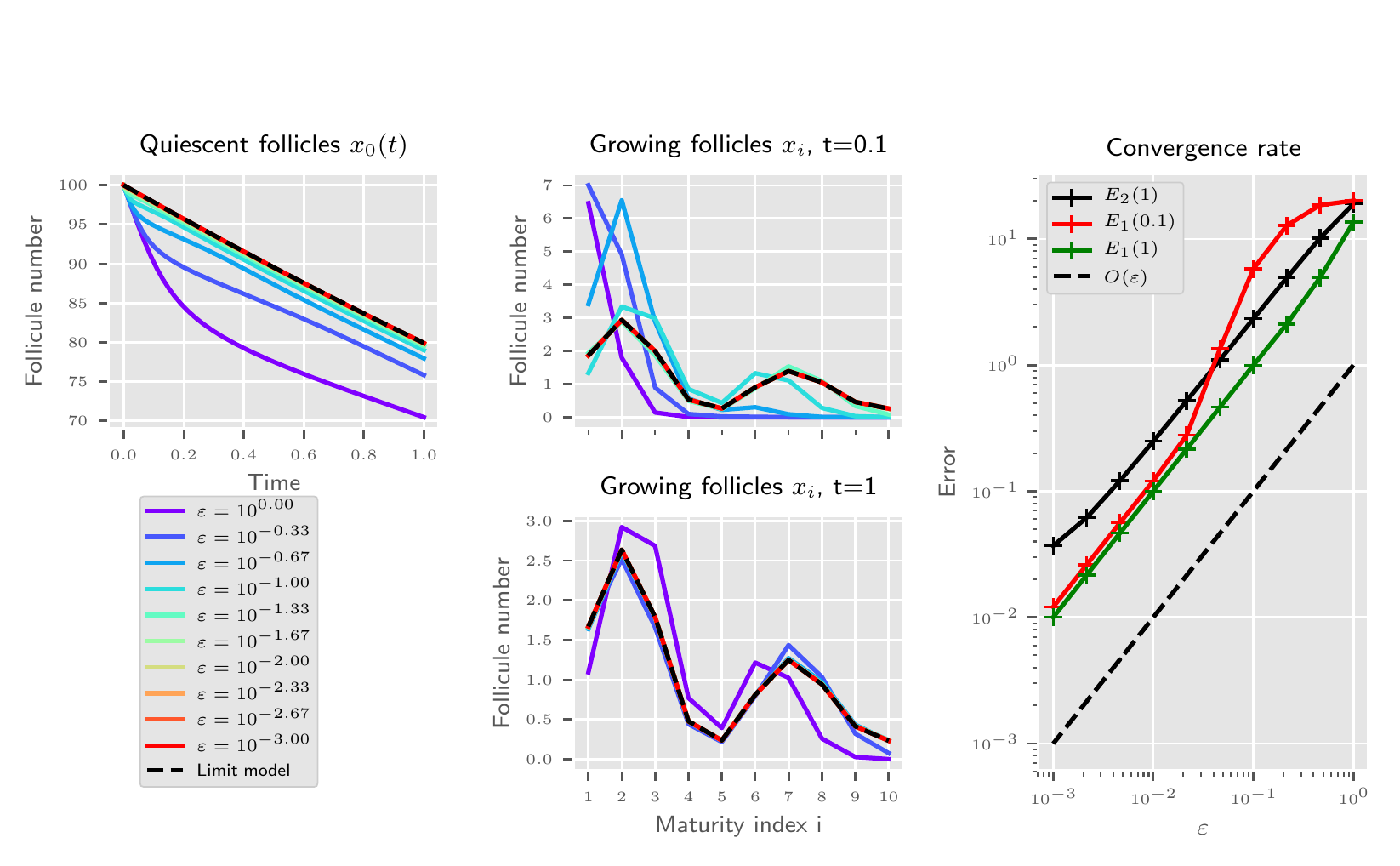}
	\caption{Trajectories \modF{in} the quiescent follicle compartment (\modF{top} left panel) and \modF{distribution of} the growing follicle population according to the maturity index $i$ at time $t=0.1$ (\modF{center top} panel) and time $t=1$ (\modF{center bottom} panel), for the rescaled variables $x_i^\varepsilon$, for different $\varepsilon$ (\modF{solid} colored lines, see legend \modF{insert}) and the reduced limit \modF{variables} $\bar x_i$ (\modF{black} dashed lines). On the right panel, we plot the discrete $l^1$ norm error $E_1(t)$ at a \modF{fixed time} $t=0.1$ and $t=1$ (\modF{solid} red and green lines, resp.) and the $l^1$-cumulative error $E_2(1)$ on $t\in(0,1)$ (black \modF{solid} line) as a function of $\varepsilon$ (see details in the text). The black dashed line is the \modF{straight line of slope 1 according} to $\varepsilon$. \modR{See the Appendix (section \ref{sec:annex}) for details on the parameter values used in the numerical simulations.}}
	\label{fig:ode_lim_mod_hist}
\end{figure}

In Figure \ref{fig:ode_lim_mod}, we plot the \modF{trajectories $x_i(t)$ in each maturity compartment} ($d=10$) for the rescaled \modF{and limit models} on a time horizon $t\in(0,1)$. \modF{In} each compartment, the trajectories of the rescaled model get closer and closer to the limit model as $\varepsilon\to0$. For $\varepsilon=0.001$, \modF{they} are almost indistinguishable. Note however that, for the growing follicles, the initial conditions of the rescaled and limit models are \modF{different}, and the convergence holds only for positive \modF{times}.

In Figure \ref{fig:ode_lim_mod_hist}, \modF{using} the same parameters as in Figure \ref{fig:ode_lim_mod}, we \modF{display the maturity distribution in} the growing follicle population, for various $\varepsilon$. \modR{We can expect from Figure \ref{fig:ode_lim_mod} that  the convergence gets} \modF{better} \modR{for larger times. We confirm this fact in Figure \ref{fig:ode_lim_mod_hist} where we compare the maturity distribution in the growing follicle population at two times, $t=0.1$ and $t=1$.}
\modR{We can quantify the error between the rescaled model and the limit model by} computing the $l^1$ error at time t,
\begin{equation}
E_1(t)= \sum_{i=0}^d \mid X_i^{\varepsilon}(t)-\bar X_i(t) \mid\,,
\end{equation}
and the cumulative error on time interval $[0,T]$,
\begin{equation}
E_2(T)= \int_0^T \sum_{i=0}^d \mid X_i^{\varepsilon}(t)-\bar X_i(t) \mid dt\,,
\end{equation}
\modF{which can be assessed numerically as}
\begin{equation}
\tilde E_2(T)= \sum_{k=0}^{N_T} \delta_t \sum_{i=0}^d \mid X_i^{\varepsilon}(t_k)-\bar X_i(t_k) \mid\,,
\end{equation}
where $t_k=k\delta_t$, for $k=0\modF{\cdots}N_T$. \modF{From the right panel of Figure \ref{fig:ode_lim_mod_hist}, we can see that, at least for a small enough $\varepsilon$, the error is inversely proportional to  $\varepsilon$},  $E_i\approx C^{te}\varepsilon^{-1}$, \modR{where the constant pre-factor may depend on the chosen norm or the particular time $t$.}

\section{PDE model}\label{sec:pde}

	\modF{When considering a continuous maturity variable, the PDE formalism is more suited for representing the follicle population dynamics}. \modF{In this section, we skip} the rescaling procedure, which follows an analogous reasoning as that detailed in section~\ref{sec:ode}, and present directly the rescaled model.

	\subsection{Rescaled model}
	\modM{Denoting by $\rho_0(t)$ the} \modR{number} \modM{of quiescent follicles and by $\rho(t,x)$ the} \modR{population} \modM{density of follicles of maturity $x$, we}
     consider the following coupled ODE-PDE system, for all $t\geq 0$,
	
	\begin{equation}\label{eq:pde_rescaled}
	\left\lbrace \begin{array}{lll}
	\ds \dfrac{d\rho_0(t)}{dt}&\ds =& \ds-(\lambda_0(\rho(t,.)) + \mu_0(\rho(t,.)))\rho_0(t)\,,\\
	\ds \varepsilon \partial_t\rho(t,x)&\ds =& \ds -\partial_x(\lambda(\rho(t,.),x)\rho(t,x))-\mu(\rho(t,.),x)\rho(t,x)\,, \quad\mbox{for } x\in(0,1)\,,\\
	\ds \lim_{x\rightarrow 0}\lambda(\rho(t,.),x)\rho(t,x)&\ds =& \ds \lambda_0(\rho(t,.))\rho_0(t)\,,\\
	\end{array}\right.
	\end{equation}
	where
	\begin{equation}\label{lambda0_mu0}
	\lambda_0(\rho(t,.))=\dfrac{f_0}{1+K_{1,0}\int_0^1 a(y)\rho(t,y)dy}\,,\quad \mu_0(\rho(t,.))=g_0 \left(1+K_{2,0}\int_0^1 b(y)\rho(t,y)dy\right)\,,
	\end{equation}
	and, for all $x\in(0,1)$,
	\begin{equation}\label{lambda_mu}
	\lambda(\rho(t,.),x)=\dfrac{f(x)}{1+K_1(x)\int_0^1 \omega_1(y)\rho(t,y)dy}\,,\quad \mu(\rho(t,.),x)=g(x)\left(1+K_2(x)\int_0^1 \omega_2(y)\rho(t,y)dy\right)
	\end{equation}
	with initial condition
	\begin{equation}\label{init_edp}
	\rho_0(t=0)=\rho_0^{in}\,,\quad \rho(t=0,x)=\rho^{in}(x)\,,\quad x\in(0,1)\,.
	\end{equation}
	We assume \modF{that} $f,g,K_1,K_2,a,b,w_1,w_2,\rho^{in}$ are regular enough functions, and will admit 
	\modF{the} existence and uniqueness of \modF{solutions} of system \eqref{eq:pde_rescaled}.
	\modF{A standard} fixed point argument, based on the mild formulation, \modF{could} be used (see for instance \cite{Collet2000,Evers2015,Evers2016} and references therein), yet this is \modF{beyond the scope of this work}. We can write the following conservation law, that gives (at least formally)
	
	\begin{equation}\label{eq:pde_rescaled_conservation}
	\frac{d }{dt} \rho_0(t) + \gge\frac{d }{dt}\int_0^1\rho(t,x) dx = - \mu_0(\rho(t,.))\rho_0(t)-\int_0^1 \mu(\rho(t,.),x)\rho(t,x)dx- \lim_{x\rightarrow 1}\lambda(\rho(t,.),x)\rho(t,x)\leq 0\,.
	\end{equation}
	
	In the following, we consider a sequence $(\rho_0^\gge, \rho^\gge)$ of \modF{solutions} of system \eqref{eq:pde_rescaled} in the limit $\gge\to 0$.

	\subsection{\modF{Limit} model}
	
	Formally, setting $\gge=0$ in system \eqref{eq:pde_rescaled} leads to the following system:
	for all $t\geq 0$,
	\begin{equation}\label{eq:pde_limit}
	\left\lbrace \begin{array}{lll}
	\ds \partial_t \bar \rho_0(t) &\ds =& \ds - ( \lambda_0(\bar \rho(t,.)) + \mu_0(\bar \rho(t,.))) \bar \rho_0(t)\,, \\
	\ds \partial_x \big( \lambda(\bar \rho(t,.),x)\bar\rho(t,x) \big) &\ds =& \ds - \mu(\bar\rho(t,.)) \bar\rho(t,x)\,,\quad\mbox{for } x\in(0,1)\,,\\
	\ds \lim_{x\rightarrow 0} \lambda(\bar\rho(t,.),x)\bar\rho(t,x) &\ds =&\ds \lambda_0(\bar\rho(t,.)) \bar\rho_0(t)\,,
	\end{array}\right.
	\end{equation}
	with \modF{an} initial condition given by $\bar \rho_0(t=0)=\rho_0^{in}$, and  $\bar \rho(t=0,.)=\bar \rho^{in}$, that satisfies the second and third lines of Eq.~\eqref{eq:pde_limit} at $t=0$.
	System \eqref{eq:pde_limit} is not necessarily well-posed, as there may be several solutions $\bar \rho $ for a given $\bar \rho_0$. \modF{In the next two specific examples, we can prove that system (\ref{eq:pde_limit}) does admit a single positive solution, which is a natural} \modF{limit candidate} for the sequence $(\rho_0^\gge, \rho^\gge)$. 
	
	\begin{example}[Linear case]\label{ex:pde_limit_linear}
		Let us suppose that $K_{1,0}=K_{2,0}=0$ and $K_1\equiv 0$, $K_2\equiv 0$.  Assume furthermore that $f(0)>0$.  Then, system \eqref{eq:pde_limit} becomes linear, and has a unique solution given by
		\begin{equation}\label{eq:pde_limit_linear}
		\left\lbrace \begin{array}{lll}
		\ds\bar \rho_0(t) &\ds=&\ds \rho_0^{in}\exp\left(-(f_0+g_0)t\right)\\
		\ds \bar \rho(t,x) &\ds=& \ds \frac{f_0}{f(0)} \bar \rho_0(t) e^{-\int_0^x \frac{g(y)+ f'(y)}{f(y)}dy}\,,\quad\mbox{for } x\in(0,1)\,.
		\end{array}\right.
		\end{equation}
	\end{example}
	
	\begin{example}[Feedback \modF{onto} quiescent follicle activation and death rates]\label{example_4}
		Let \modF{us} suppose that $K_{1}\equiv 0$, $K_{2}\equiv 0$, and $K_{1,0}>0$ and $K_{2,0}\geq 0$. Assume furthermore that $f(0)>0$. Then, system \eqref{eq:pde_limit} \modR{with positivity requirement} $\bar{\rho}(t,x)\geq 0$ \modF{can be rewritten as}
		\begin{equation}\label{eq:pde_limit_K_1_0}
		\left\lbrace \begin{array}{lll}
		\ds\frac{d \bar \rho_0(t)}{dt} &\ds=&\ds -\big(\gl_0(\bar \rho(t,.)) + \mu_0(\bar \rho(t,.)) \big) \bar \rho_0(t) \\
		\ds \bar  \rho(t,0) &\ds=& \ds \frac{-f(0)+\sqrt{(f(0))^2 +4f_0\bar \rho_0(t)f(0)K_{1,0}\int_0^1 a(x)e^{\int_0^x \frac{g(y)+ f'(y)}{f(y)}dy}dx }}{2f(0)K_{1,0}\int_0^1 a(x)e^{\int_0^x \frac{g(y)+ f'(y)}{f(y)}dy}dx}\\
		\ds \bar \rho(t,x) &\ds=& \ds \bar \rho(t,0) e^{\int_0^x \frac{g(y)+ f'(y)}{f(y)}dy}\,,\quad\mbox{for } x\in(0,1)\,,
		\end{array}\right.
		\end{equation}
	\modF{which admits a unique solution.}
		Indeed, the functional expression of $\bar  \rho$ comes directly from solving the second line of Eq.~\eqref{eq:pde_limit}. Using the boundary condition in the third line of Eq.~\eqref{eq:pde_limit}, one can verify that $\bar  \rho(t,0)$ is the positive root of a polynomial of degree $2$, namely
		\[f_0\bar  \rho_0(t)=f(0) \bar\rho(t,0)\left(1+K_{1,0}\bar\rho(t,0)\int_0^1 a(x)e^{\int_0^x \frac{g(y)+ f'(y)}{f(y)}dy}dx\right)\]
	    \modF{The} limit system \eqref{eq:pde_limit_K_1_0} is \modF{a} nonlinear ODE. To simplify notations, we introduce $$
	H(x)=\displaystyle e^{-\displaystyle\int_0^x \frac{g(y)+ f'(y)}{f(y)}dy},\quad
	H_a=\displaystyle\int_0^1 a(x)H(x)dx, \quad H_b=\displaystyle\int_0^1 b(x)H(x)dx$$
	from which we rewrite \modF{system \eqref{eq:pde_limit_K_1_0}} as
	$$
	\left\lbrace \begin{array}{lll}
	\displaystyle\frac{d \bar \rho_0(t)}{dt} &\displaystyle=&\displaystyle -\big(\lambda_0(\bar \rho(t,.)) + \mu_0(\bar \rho(t,.)) \big) \bar \rho_0(t) \\
	\displaystyle\bar  \rho(t,0) 
	&\displaystyle=& \displaystyle \frac{-1+\sqrt{1 +4\bar \rho_0(t)K_{1,0}H_a\dfrac{f_0}{f(0)}}}{2K_{1,0}H_a}\\
	\displaystyle \bar \rho(t,x) &\displaystyle=& \displaystyle \bar \rho(t,0) H(x)\,,\quad\mbox{for } x\in(0,1)\,.\\
	\lambda_0(\rho(t,.))&\displaystyle=&\dfrac{f_0}{1+K_{1,0}\bar  \rho(t,0)H_a}\,.\\
	\mu_0(\rho(t,.))&\displaystyle=&g_0\left(1+K_{2,0}\bar  \rho(t,0)H_b\right)
	\end{array}\right.
	$$
	We can thus solve $\bar \rho_0(t)$ as the solution of an \modR{autonomous} nonlinear \modR{ODE}
	\begin{equation}\label{eq:pde_limit_ex4}
 \dfrac{d}{dt} 	\bar\rho_0(t)=G(\bar \rho_0(t)),\quad \bar \rho_0(0)=\rho_0^{ini}
	\end{equation}
	with 
	$$
	G({\color{red}{\rho}})=-\left[\dfrac{2f_0}{1+\displaystyle \sqrt{1 +4\dfrac{f_0}{f(0)}K_{1,0}H_a{\color{red}{\rho}}}}   + g_0\left(1+K_{2,0}\dfrac{-1+\sqrt{1+4\dfrac{f_0}{f(0)}K_{1,0}H_a  {\color{red}{\rho}}  }}{2K_{1,0}H_a}H_b\right)\right]{\color{red}{\rho}}
	$$
	from which we then compute $\bar  \rho(t,0)$ and eventually $\bar  \rho(t,x)$.
	
	\end{example}
	
	\subsection{Convergence in the linear case}\label{lin_rescaled}
	In this paragraph, we assume that $K_{1,0}=K_{2,0}=0$, $K_1\equiv 0$, $K_2\equiv 0$, and $f(0)>0$ as in Example \ref{ex:pde_limit_linear}. In such a case, one can solve explicitly system \eqref{eq:pde_rescaled} for each $\gge>0$ using the characteristics method. We obtain
	\begin{equation}\label{eq:pde_rescaled_solution_linear}
	\left\lbrace \begin{array}{lll}
	\ds \rho_0^{\gge}(t) &\ds=&\ds \rho_0^{in}\exp\left(-(f_0+g_0)t\right)\,, \\
	\ds  \rho^{\gge}(t,x)  &\ds=& \ds  \begin{cases}
	e^{-\int^x_{X(0;t,x)} \frac{g(y)+f'(y)}{f(y)}dy} \rho^{ini}(X(0;t,x))\,,\quad\mbox{if } t \leq \int^x_0\frac{\gge}{f(y)}dy\,,\\
	\frac{f_0}{f(0)} \rho_0^{in} \, e^{-(f_0+g_0)\,(t - \int^x_0\frac{\gge}{f(y)}dy)} \, e^{-\int^x_0 \frac{g(y)+f'(y)}{f(y)}dy}\,,\quad\mbox{if } t > \int^x_0\frac{\gge}{f(y)}dy\,.
	\end{cases}
	\end{array}\right.
	\end{equation}
	where $X(0;t,x)$ is the location of the characteristic at time $0$, given that it goes through the point $x$ at time $t$, namely:
	\begin{equation}\label{eq:charact}
	\frac{d}{ds}X(s;t,x)= f\left(X(s;t,x)\right)\,,\quad X(t;t,x)=x\,.
	\end{equation}
	It is thus clear that $\rho_0^{\gge}=\bar {\rho}_0$ is a constant sequence \modF{in $\varepsilon$}. \modR{For the population density $\rho^{\gge}$}, which is here the fast \modF{unknown}, we prove the following
	\begin{prop}\label{prop:pde_cv_linear}Assume that $K_{1,0}=K_{2,0}=0$, $K_1\equiv 0$, $K_2\equiv 0$, and $f(0)>0$ with $\int_0^1 \frac{1}{f(x)}dx<\infty$ and $\int^1_0 \frac{g(x)+f'(x)}{f(x)}dx<\infty$. Then, for all $\eta>0$, we have
		\begin{equation}
		\lim_{\gge \to 0} \sup_{t  >\eta} \sup_{x \in (0,1)} \vert \rho^{\gge}(t,x) -\bar \rho(t,x) \vert = 0
		\end{equation}	
		where $\bar \rho$ is given in Eq.~\eqref{eq:pde_limit_linear}.
	\end{prop}
	\begin{proof}
		It is clear that for any $\eta>0$, there exists $\gge'$ such that for all $\gge<\gge'$, and all $t>\eta$ we have
		\[t> \int^x_0\frac{\gge}{f(y)}dy \]
		Then, comparing the solutions of Eq.~\eqref{eq:pde_rescaled_solution_linear} and  Eq.~\eqref{eq:pde_limit_linear} allows us to conclude that, for all $t>\eta$,
		\[  \sup_{t  >\eta} \sup_{x \in (0,1)} \vert \rho^{\gge}(t,x) -\bar \rho(t,x) \vert \leq  \frac{f_0}{f(0)} \rho_0^{in}\left(\sup_{x\in(0,1)}e^{-\int^x_0 \frac{g(y)+f'(y)}{f(y)}dy}\right) \left(e^{(f_0+g_0)\gge\int^1_0\frac{1}{f(y)}dy} -1 \right)\to 0\,,\quad as\quad \gge \to 0\,.\]
	\end{proof}	
	\begin{rem}\label{rem:pde}
		As in Remark \ref{rem:ode}, we can see that during a time of order $\gge t$, we cannot get \modF{the} convergence of the rescaled model towards the reduced one, \modF{which precludes} uniform convergence in time starting from $\modF{t=}0$.
	\end{rem}

	\subsection{Numerical study}
	
	In this paragraph, we detail the numerical schemes that we have designed to solve both systems \eqref{eq:pde_rescaled} and \eqref{eq:pde_limit}, and illustrate the consistency and convergence of these algorithms using the exact solutions.
	
	\subsubsection{Numerical scheme for the \modF{limit} model}\label{ssec:algo_limit_model}
	\modF{We} design a finite difference scheme to compute a numerical solution \modF{to} the PDE limit system~\eqref{eq:pde_limit}.
	This system is nonlinear due to the dependence of $\lambda_0$, $\lambda$ and $\mu$ upon the solution   $\bar\rho(t,x)$, which itself depends on $\bar \rho_0$. We propose to treat this \modF{nonlinearity} with a fixed point scheme.  
	At each time step $t_n=n\Delta_t$, for $n=0,\ldots,N$ with  $T=N\Delta_t$ we build a sequence $\rho^\ell(t_n,x)$  such that $$\lim_{\ell\rightarrow\infty}\rho^\ell(t_n,x) =\bar\rho(t_n,x).$$
	Let $x_k=k\Delta_x$,  for $k=0,\ldots,M$, with $M\Delta_x=1$. We \modF{introduce} the discretized approximations  $$\rho^{n,\ell}_k\approx\rho^\ell(t_n,x_k),\quad \bar\rho^{n}_k\approx\lim_{\ell\rightarrow\infty}\rho^\ell(t_n,x_k)\quad\mbox{and}\quad \bar{\bar\rho}_0^{n}\approx \bar\rho_0(t_n)\,,$$
	which we compute as follows. Let $\eta\ll1$.
	\begin{enumerate}
		\item Initialization.
		$$\rho^{0,0}_k=\bar \rho^{in}(x_k)\quad\mbox{and }\quad \bar{\bar\rho}_0^{0}=\rho^{in}_0.$$
		\item For $n=0\nearrow N$ compute  $\bar\rho_k^{n}=\lim_{\ell\rightarrow\infty}\rho_k^{\ell,n}$ iteratively then update $\bar{\bar\rho}_0^{n+1}$ as follows.
		\begin{enumerate}
			\item Initialize residual $R_\ell^n=1$ and set  $\ell=0$
			\item While $R_\ell^n>\eta$ do 
			\begin{itemize}
				\item Compute \modR{the PDE parameters}  $\bar\lambda_0^{n,\ell}$, $\lambda_k^{n,\ell}$ and $\mu_k^{n,\ell}$ (for $k=0,\ldots,M$) \modR{by standard trapezoidal rules} 
				\begin{equation*}
				\begin{array}{rcl}
				\displaystyle \bar\lambda_0^{n,\ell} & \displaystyle = & \displaystyle \frac{f_0}{1+K_{1,0}\Delta_x\left(\frac{1}{2}a(0)\rho_0^{n,\ell}+\sum_{j=1}^{M-1}a(x_j)\rho_j^{n,\ell}+\frac{1}{2}a(x_M)\rho_M^{n,\ell}\right)}\\
				\displaystyle \lambda_k^{n,\ell} & \displaystyle = & \displaystyle \frac{f(x_k)}{1+K_1(x_k)\Delta_x\left(\frac{1}{2}\omega_1(0)\rho_0^{n,\ell}+\sum_{j=1}^{M-1}\omega_1(x_j)\rho_j^{n,\ell}+\frac{1}{2}\omega_1(x_M)\rho_M^{n,\ell}\right)}\\
				\displaystyle \mu_k^{n,\ell} & \displaystyle = & \displaystyle g(x_k)\left(1+K_2(x_k)\Delta_x\left(\frac{1}{2}\omega_2(0)\rho_0^{n,\ell}+\sum_{j=1}^{M-1}\omega_2(x_j)\rho_j^{n,\ell}+\frac{1}{2}\omega_2(x_M)\rho_M^{n,\ell}\right)\right)
				\end{array}
				\end{equation*}
				\item \modM{Enforce} \modF{boundary} condition at $x=0$ 
				$$\lambda_0^{n,\ell}\rho_{0}^{n,\ell+1}=\bar\lambda_0^{n,\ell}\bar{\bar\rho}_0^{n}$$
				\item \modF{Integrate numerically the PDE} in $x$ 
				$$\lambda_k^{n,\ell}\left(\rho_{k+1}^{n,\ell+1}-\rho_k^{n,\ell+1}\right)=-\Delta_x\mu_k^{n,\ell}\rho_k^{n,\ell+1},\quad k=0,\ldots,M-1$$
				\item Compute residual between $\ell$ and $\ell+1$ iterations
				$$ R^n_\ell=\max_{k=0,\ldots,M}|\rho_k^{n,\ell+1}-\rho_k^{n,\ell}|$$
			\end{itemize}
			\item Set \modR{the new population density equal to the value obtained at the end of the fixed point procedure}
			$\bar\rho_k^{n}=\rho_k^{n,\ell+1}$ 
			\item Compute \modR{the ODE parameters from the fixed point value, using standard trapezoidal rules}
			\begin{equation*}
				\begin{array}{rcl}
				\displaystyle \bar{\bar\lambda}_0^{n} & \displaystyle = & \displaystyle \frac{f_0}{1+K_{1,0}\Delta_x\left(\frac{1}{2}a(0)\bar\rho_0^{n}+\sum_{j=1}^{M-1}a(x_j)\bar\rho_j^{n}+\frac{1}{2}a(x_M)\bar\rho_M^{n}\right)}\\
				\displaystyle \bar{\bar\mu}_0^{n} & \displaystyle = & \displaystyle g_0\left(1+K_{2,0}\Delta_x\left(\frac{1}{2}b(0)\bar\rho_0^{n}+\sum_{j=1}^{M-1}b(x_j)\bar\rho_j^{n}+\frac{1}{2}b(x_M)\bar\rho_M^{n}\right)\right)
			    \end{array}
			\end{equation*}
			\item \modF{Integrate numerically the ODE} $\bar{\bar\rho}_0$  between $t_n$ and $t_{n+1}$ \modR{with a classic explicit Euler scheme}
			$$\bar{\bar\rho}_0^{n+1}=\bar{\bar\rho}_0^{n}-\Delta_t \left(\bar{\bar\lambda}_0^{n}+\bar{\bar\mu}_0^{n}\right)\bar{\bar\rho}_0^{n}.$$
		\end{enumerate}
	\end{enumerate}
	
	\subsubsection{Convergence of the numerical scheme for the \modF{limit} model.}\label{test_case_red}
	\begin{figure}[h]
		\includegraphics[width=\textwidth]{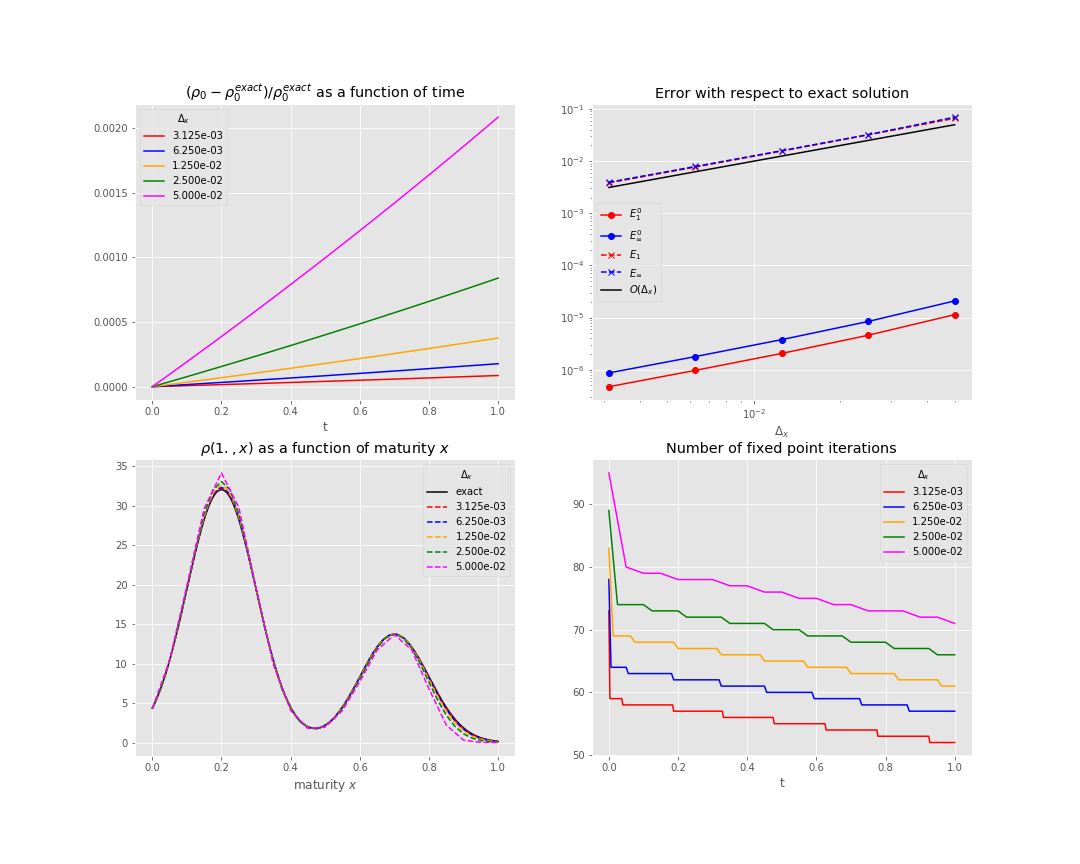}
		\caption{Convergence of the numerical scheme for the \modF{limit} model as a function of $\Delta_x=\Delta_t$. Top left \modF{panel:} relative error of $\bar {\bar\rho}_0^n\approx\rho_0(t_n)$, computed by means of the fixed point algorithm, compared the numerical solution of \eqref{eq:pde_limit_ex4}, computed with the odeint python ODE solver. \modF{Bottom left panel:} comparison of $\bar\rho_k^{N}=\bar \rho(t=1,x_k)$ and the numerical solution of \eqref{eq:pde_limit_K_1_0}. Top right panel: \modM{relative} errors with respect to the pseudo exact solution of \eqref{eq:pde_limit_K_1_0}. \modF{Bottom right panel}: number of iterations in the fixed point algorithm. 
		}
		\label{fig:red_mod}
	\end{figure}
	\modRR{For the nonlinear scenario described in Example \ref{example_4} we obtain a ``pseudo exact'' solution using scipy library ODE solver odeint to solve Eq.~\eqref{eq:pde_limit_ex4} and compute $\bar\rho_0(t)$, and we use Eq.~\eqref{eq:pde_limit_K_1_0} to compute $\bar\rho(t,\cdot)$. We use this reference solution to assess the performances of the numerical scheme described in paragraph \ref{ssec:algo_limit_model}.}
\modMM{We display on Figure \ref{fig:red_mod} the simulations performed with the parameter values detailed in the Appendix} (section \ref{sec:annex}),  on a time horizon $t\in(0,1)$, and an increasingly fine discretization $$\Delta_x\in\{0.05,0.025,0.0125,0.00625,0.003125\}\,.$$
	The top left panel \modF{shows} the difference between the pseudo exact solution  $\bar \rho_0(t_n)$ of Eq.~\eqref{eq:pde_limit_ex4} 
	and the numerical \modF{solutions} $\bar{\bar\rho}^n$ computed with the fixed point algorithm, as a function of $t\in[0,T]$. \modF{The} bottom left panel \modF{shows} the pseudo exact solution $\bar \rho(T,x_k)$ and the numerical solutions $\bar\rho^{N}_k$ at final time $T=1$, as a function of $x\in [0,1]$. In the top right panel, we display the relative errors in $L^1$ and $L^\infty$ \modF{norms} between the pseudo exact and the numerical solution $\bar\rho_0(t)$, as a function of $\Delta_t=\Delta_x=1/N$:
	$$E^0_1(\Delta_x)=\displaystyle \frac{\sum_{n=0}^{N}\valabs{\bar \rho_0(t_n)-\bar{\bar\rho}^n}}{\sum_{n=0}^{N}\valabs{\bar \rho_0(t_n)}},$$
	$$E^0_\infty(\Delta_x)=\displaystyle \frac{\max_{n=0,\ldots,N}\valabs{\bar \rho_0(t_n)-\bar{\bar\rho}^n}
	}{\max_{n=0,\ldots,N}\valabs{\bar \rho_0(t_n)}},
	$$
	and the relative errors between the pseudo exact and the numerical solutions $\bar \rho(T,x)$
	$$E_1(\Delta_x)=\displaystyle \frac{\sum_{k=0}^{M}\valabs{\bar \rho(T,x_k)-\bar\rho^{N}_k}
	}{\sum_{k=0}^{M}\valabs{\bar \rho(T,x_k)}},
	$$
	$$E_\infty(\Delta_x)=\displaystyle \frac{\max_{k=0,\ldots,M}\valabs{\bar \rho(T,x_k)-\bar\rho^{N}_k}}{
	\max_{k=0,\ldots,M}\valabs{\bar \rho(T,x_k)}}.
	$$
	As expected, the errors are linear in $\Delta_x$, which means that the \modF{order of numerical convergence} is one.
 In the bottom right panel, we display the number of iterations \modF{performed} in the inner loop of \modF{the} fixed point algorithm, as a function of time. The number of iterations \modF{needed} to converge decreases with time and the number of time points. \modF{This tendency is not really surprising, since,} at each time step, we start with the solution \modF{obtained} at the previous time step \modF{as} initial condition \modM{for the fixed point loop}. Since the solution decreases with time, the distance between the \modM{fixed point} initial condition and the solution decreases with  \modF{both the time and time step, hence} convergence requires less iterations. 
 
	\subsubsection{Numerical scheme for the rescaled model}
	\modF{We} design an explicit finite volume scheme to compute a numerical solution \modF{to} the rescaled model (Eq. \ref{eq:pde_rescaled}).
	The discretized unknowns are at each time step $t_n=n\Delta_t$, for $n=0,\ldots,N$ with  $T=N\Delta_t$ 
	$${\bar\rho}_0^{\varepsilon,n}\approx \rho_0^\varepsilon(t_n)\,,$$
	and, for $x_k=k\Delta_x$, $k=0,\ldots,M$, with $M\Delta_x=1$, 
	$$\rho^{\varepsilon,n}_k\approx\int_{x_k}^{x_{k+1}}\rho^\varepsilon(t_n,x)dx,\quad k=0,\ldots,M-1\,.$$
	
	\modMM{We integrate numerically the PDE between $t_n$ and $t_{n+1}$ and over $[{x_k},{x_{k+1}}]$ by freezing the nonlinear coefficients  $\lambda(\bar\rho(t,.),x)$ and $\mu(\bar\rho(t,.))$ at time  $t_n$}
		\begin{equation}\label{eq:explicit_nosplit}\rho_{k}^{\varepsilon,n+1}-\rho_k^{\varepsilon,n}=-
			\frac{\Delta_t^n}{\varepsilon\Delta_x}\left(\lambda_k^{\varepsilon,n}\rho_{k}^{\varepsilon,n}-\lambda_{k-1}^{\varepsilon,n}\rho_{k-1}^{\varepsilon,n}\right)-\frac{\Delta_t^n}{\varepsilon}\mu_k^{\varepsilon,n}\rho_k^{\varepsilon,n},\quad k=1,\ldots,M\,.\end{equation}
	
	At each time step, we compute both the PDE and ODE coefficients using the midpoint rule and the \modMM{numerical} solution $(\rho^{\varepsilon,n}_k)_{k=0,\ldots,M}$ \modMM{as a piecewise constant solution}
	\begin{eqnarray}\label{coeff_eps}\left\{\begin{array}{l}
	\lambda_k^{\varepsilon,n}=\displaystyle\frac{f(x_{k+1/2})}{1+K_1(x_{k+1/2})\Delta_x\sum_{j=0}^{M-1}\omega_1(x_{j+1/2})\rho_j^{\varepsilon,n}},\\
	\bar\lambda_0^{\varepsilon,n}=\displaystyle\frac{f_0}{1+K_{1,0}\Delta_x\sum_{j=0}^{M-1}a(x_{j+1/2})\rho_j^{\varepsilon,n}},\\
	\mu_k^{\varepsilon,n}=g(x_{k+1/2})\left(1+K_2(x_{k+1/2})\Delta_x\sum_{j=0}^{M-1}\omega_2(x_{j+1/2})\rho_j^{\varepsilon,n}\right),\\
	\bar\mu_0^{\varepsilon,n}=g_0\left({1+K_{2,0}\Delta_x\sum_{j=0}^{M-1}b(x_{j+1/2})\rho_j^{\varepsilon,n}}\right).\end{array}\right.\end{eqnarray}
	
	\modMM{For the explicit scheme \eqref{eq:explicit_nosplit}, two stability conditions must be \modF{satisfied}} 
	\begin{itemize}\item CFL-like stability condition:
		$$\frac{\Delta_t^n}{\varepsilon\Delta_x}\lambda_{k}^{\varepsilon,n}\leq C_{cfl}<1,\quad k=1,\ldots,M\,,$$
		which can be rewritten as
		\begin{equation}\label{eq:stab_pde_1}
		\frac{\Delta_t^n}{\varepsilon\Delta_x}\leq \displaystyle\frac{C_{cfl}}{\max_k \lambda_{k}^{\varepsilon,n}}\,.
		\end{equation}
		\item Positivity conservation condition:
		$1+\frac{\Delta_t^n}{\varepsilon}\left(\frac{\lambda_{k}^{\varepsilon,n}}{\Delta_x}-\mu_{k}^{\varepsilon,n}\right)\geq 0,\quad k=1,\ldots,M$
		we impose that if  $\rho_i^{\varepsilon,n}=\delta_{ik}$ then $\rho_i^{\varepsilon,n+1}\geq 0$ for all $k$ and all $i$ which leads to
		$$1-\frac{\Delta_t^n}{\varepsilon}\left(\frac{\lambda_{k}^{\varepsilon,n}}{\Delta_x}+\mu_{k}^{\varepsilon,n}\right)\geq 0,\quad k=1,\ldots,M\,,$$
        which can be rewritten as
		\begin{equation}\label{eq:stab_pde_2}
		\frac{\Delta_t^n}{\varepsilon\Delta_x}\leq\displaystyle\frac{1}{\max_{k}\left(\lambda_{k}^{\varepsilon,n}+\Delta_x \mu_{k}^{\varepsilon,n}\right)}\,.
		\end{equation}
		\end{itemize}

	\modMM{The overall numerical scheme proceeds as follows:}
	\begin{enumerate}
		\item Initialization :
		$$\rho^{\varepsilon,0}_k=\rho^{in}(x_k)\quad\mbox{and }\quad \bar \rho_0^{\varepsilon,0}=\rho^{in}_0.$$
		\item For $n=0\nearrow N$ compute  $\rho_k^{\varepsilon,n+1}$ for $k=1,\ldots,M$ then update $\bar\rho_0^{\varepsilon,n+1}$ \modR{and finally compute} $\rho_0^{\varepsilon,n+1}$ :
		\begin{enumerate}
			\item Compute the PDE and ODE coefficients $\bar\lambda_0^{\varepsilon,n}$, $\bar\mu_0^{\varepsilon,n}$,  $\lambda_k^{\varepsilon,n}$ and $\mu_k^{\varepsilon,n}$ (Eq. \ref{coeff_eps})
			\item Compute  $\Delta_t^n$ satisfying stability conditions \eqref{eq:stab_pde_1}  and  \eqref{eq:stab_pde_2}
			\item \modF{Integrate numerically the PDE} \modMM{in $x$ at time  $t_n$ using \eqref{eq:explicit_nosplit}}
			\item \modF{Integrate numerically the ODE} ${\bar\rho}_0$  between $t_n$ and $t_{n+1}$
			$${\bar\rho}_0^{\varepsilon,n+1}={\bar\rho}_0^{\varepsilon,n}-\Delta_t^n\displaystyle\left({\bar\lambda}_0^{\varepsilon,n}+{\bar\mu}_0^{\varepsilon,n}\right).$$
			\item \modM{Enforce} \modF{the boundary} condition  at $x=0$ 
			$$\lambda_0^{\varepsilon,n+1}\rho_{0}^{\varepsilon,n+1}=\bar\lambda_0^{\varepsilon,n+1}\bar\rho_0^{\varepsilon,n+1}\,.$$
		\end{enumerate}
	\end{enumerate}

\subsubsection{Convergence of the numerical scheme for the rescaled model.}
We start by checking the convergence of the numerical scheme in a case where we know an exact solution, that is the linear case $K_{10}=K_{20}=K_1=K_2=0$. We test \modF{several} discretizations \modM{$$\Delta_x\in\{0.1,0.05,0.025,0.0125,0.00625,0.003125\},$$}
for \modFF{the parameter values of the linear} scenario detailed in \modF{the Appendix} (section \ref{sec:annex}).
The results are displayed in Figure \ref{fig:rescaled_mod_xt_conv}. In the left panel, we see that $\rho_0(t)$ is computed exactly since all curves corresponding to different discretizations are superimposed \modR{(as expected in the linear case)}. In the center panel, we display the solution at a fixed time $T$ as a function of $x$, which does depend on the discretization. The right panel shows the relative error curves, which exhibit a \modR{convergence rate} better than linear.

\begin{figure}[h]
		\includegraphics[width=\textwidth]{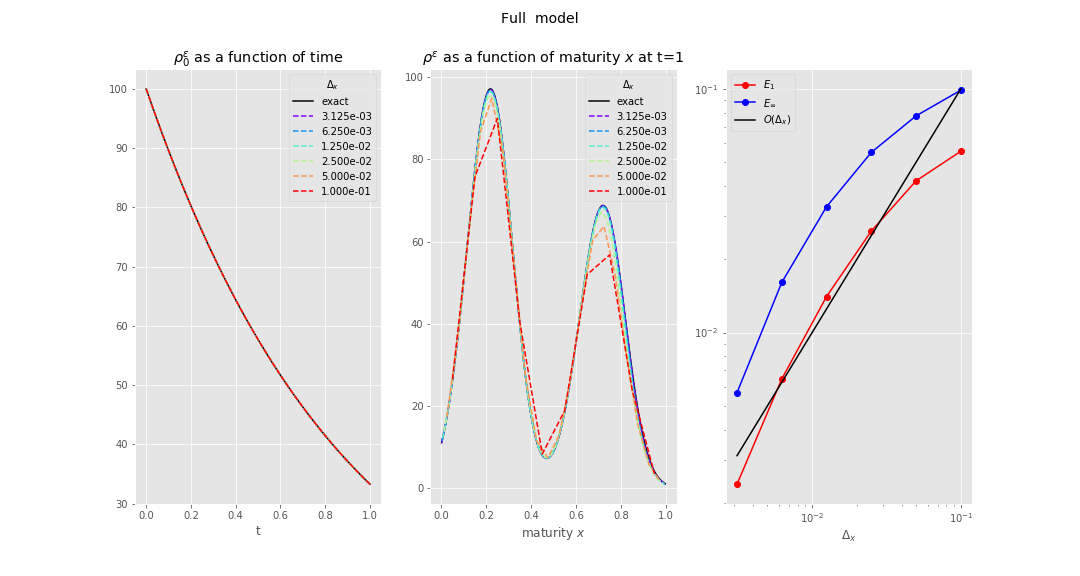}
		\caption{Convergence of the numerical scheme for the rescaled model as a function of $\Delta_x$, for $\varepsilon=0.5$.
			Left panel: $\rho_0^{\varepsilon}$ as a function of time (exact and finite volume scheme),  center panel: $\rho^{\varepsilon}$ as a function of $x$ at final time $t=1$, right panel: $L_\infty$ and $L_1$ relative errors with respect to \modF{the} exact solution. 
			}
		\label{fig:rescaled_mod_xt_conv}
	\end{figure}

	\subsubsection{$\varepsilon$-convergence towards \modF{the limit} model}
	So far we have proved the convergence of the rescaled model towards the limit model when $\varepsilon\to 0$ in the linear case. We can only test it numerically in the general case. To minimize the numerical error arising from solving the limit model numerically, \modF{we} illustrate the $\varepsilon$-convergence of $(\rho_0^\gge, \rho^\gge)$ towards $(\bar \rho_0, \bar \rho)$ in the nonlinear scenario of example \ref{example_4} (see details \modF{on the} parameter values in \modF{the Appendix} (section \ref{sec:annex})), \modF{for which} the  pseudo-exact solution of the limit model is available. As in Figure \ref{fig:red_mod}, \modMM{pseudo exact solutions} $(\bar \rho_0, \bar \rho)$ are simulated using Eq.~\eqref{eq:pde_limit_ex4} and Eq.~\eqref{eq:pde_limit_K_1_0}. The results are displayed in Figure \ref{fig:rescaled_mod_eps_conv} for a set of \modF{$\varepsilon$ values} and two discretizations $M=N=100$  and $M=N=200$. The agreement of $\rho_0(t)$ with the \modF{limit} model solution \modF{(top left panel)}, and that of $\rho^\varepsilon(t,x)$ \modF{(bottom left panel)}, are qualitatively good as soon as $\varepsilon\leq 10^{-2}$. The relative error \modF{for} $\rho_0(t)$ \modF{(top right panel) exhibits} a linear behavior in $\varepsilon$. The error curves for the convergence of the solution $\rho^\varepsilon(t,x)$ in the domain \modF{(bottom right panel)} are not linear, and \modF{remain beyond a} threshold when $\varepsilon$ goes to 0.  However, \modF{both} the value of $\varepsilon$ \modF{for which} the error \modF{approaches} this threshold, and the value of the error itself decrease when we refine the discretization. This indicates that
    we should refine the discretization when we decrease $\varepsilon$. Since our current numerical scheme is explicit in time, \modF{any} refinement must be simultaneous in $\Delta_x$ and $\Delta_t$, \modF{and as a consequence}  the cost in CPU time \modF{depends} quadratically in $\varepsilon^{-1}$.
	 
	 In practice, realistic values for $\varepsilon$ should remain tractable. However this behavior is a good incentive to study a more economical numerical scheme, namely an implicit one, which would provide accurate results with coarser discretizations. 
	
	\begin{figure}[h]
	\includegraphics[width=\textwidth]{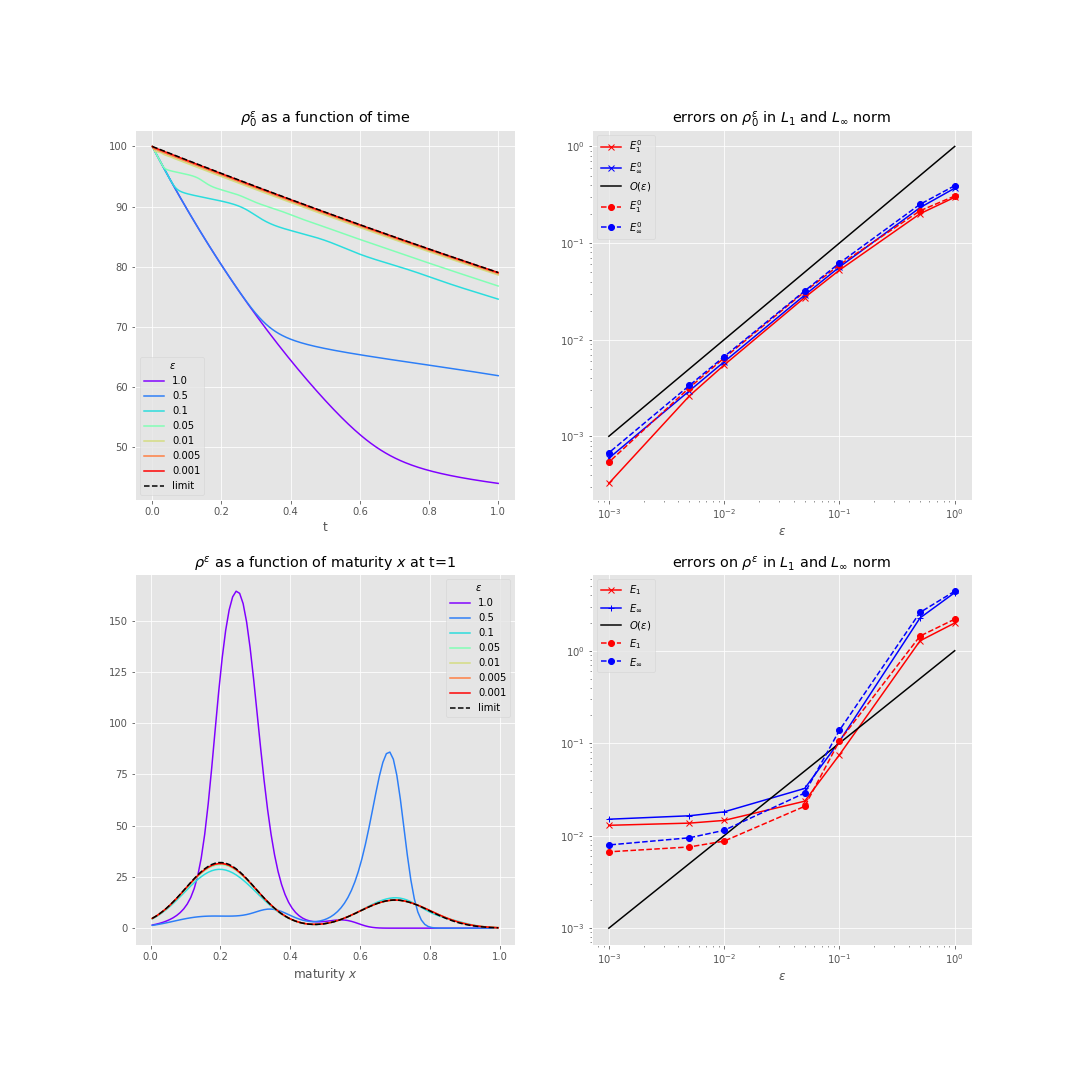}
		\caption{Numerical assessment of the $\gge$-convergence of the rescaled model towards the limit model.
			\modF{Top} left panel: $\bar\rho_0^{\varepsilon,n}$ and $\bar \rho_0$.  \modF{Bottom} left panel: $\rho_k^{\varepsilon,N}$ and $\bar\rho(1,\cdot)$. \modF{Top} right panel: relative error in $L_1$ and $L_\infty$ \modF{norms} between $\rho_0^{\gge}(t)$ and $\bar \rho_0$. \modF{Bottom} right panel: relative error in $L_1$ and $L_\infty$ \modF{norms} between $\rho^\varepsilon(1,\cdot)$ and $\bar\rho(1,\cdot)$. In both right panels, solid lines correspond to solutions computed with $M=N=100$ and \modF{dashed} lines to solutions computed with $M=N=200$. %
			}
		\label{fig:rescaled_mod_eps_conv}
	\end{figure}

 \section{SDE model}

\modF{We now turn to a stochastic model for the follicle population dynamics}. We \modF{skip} the rescaling procedure, which follows an analogous reasoning as that detailed in section~\ref{sec:ode}, and present directly the rescaled model.

\subsection{Rescaled model}

We consider the following coupled \modF{Poisson}-driven SDE system, given by, for all $t\geq 0$,

\begin{equation}\label{eq:sde_rescaled}
\left\lbrace \begin{array}{lll}
\ds X_0^\varepsilon(t)&\ds =&\ds X_0^{\varepsilon,in}-\varepsilon P_0^+\left(\int_0^t\dfrac{\lambda_0(X^\varepsilon(s))}{\varepsilon}X_0^\varepsilon(s)ds\right) \\
& &\ds -\varepsilon P_0^-\left(\int_0^t\dfrac{\mu_0(X^\varepsilon(s))}{\varepsilon}X_0^\varepsilon(s)ds\right)\\
\ds X_i^\varepsilon(t)&\ds =&\ds X_i^{\varepsilon,in}+ P_{i-1}^+\left(\int_0^t\dfrac{\lambda_{i-1}(X^\varepsilon(s))}{\varepsilon}X_{i-1}^\varepsilon(s)ds\right)\\
&&\ds - P_{i}^+\left(\int_0^t\dfrac{\lambda_i(X^\varepsilon(s))}{\varepsilon}X_i^\varepsilon(s)ds\right)\\
&&\ds - P_i^-\left(\int_0^t\dfrac{\mu_i(X^\varepsilon(s))}{\varepsilon}X_i^\varepsilon(s)ds\right)\,,\quad i \in \{1,\dots,d\} \\
\end{array}\right.
\end{equation}
where $X^\gge$ represents the vectorial process $(X_0^\varepsilon,\cdots,X_d^\varepsilon)$ of the follicle number in each maturity stage, the functions $\gl_i,\mu_i$, for $i\in \{0,\cdots,d\}$ are given by Eq.~\eqref{eq:lambdai_mui_ode} and \eqref{eq:lambda0_mu0_ode_rescaled} and $\left(P_i^+,P_i^-\right)_{i\in\{0, \cdots, d\}}$ are independent standard Poisson processes. In Eq.~\eqref{eq:sde_rescaled}, $X_0^{\varepsilon,in}$ is \modF{an} $\gge \mathbb{N}$-valued random variable, and each $X_i^{\varepsilon,in}$, $i\in \{1,\dots,d\}$, is a $\mathbb{N}$-valued random variable.

\modF{SDE (\ref{eq:sde_rescaled})} defines uniquely (in law) a continuous time Markov chain in $\gge\N\times\N^d$. We note the following conservation law,

\begin{equation}\label{eq:sde_rescaled_conservation}
X_0^\varepsilon(t) + \gge  \sum_{i=1}^d   X_i^\varepsilon(t)= X_0^\varepsilon(0) + \gge  \sum_{i=1}^d   X_i^\varepsilon(0) -\gge \sum_{i=0}^d P_i^-\left(\int_0^t\dfrac{\mu_i(X^\varepsilon(s))}{\varepsilon}X_i^\varepsilon(s)ds\right)\,.
\end{equation}

In the following, we consider a sequence $X^\gge$ of \modF{solutions to system} \eqref{eq:sde_rescaled} in the limit $\gge$ tends to $0$.

\subsection{Limit model}

Formally, setting $\gge=0$ in system \eqref{eq:sde_rescaled} \modF{leads to} the following system for $(\bar x_0,\bar f_t)$, \modF{coupling the dynamics of a deterministic continuous function on $\mathbb{R}_+$, $\bar x_0$, with those of a time-dependent measure on $\mathbb{N}^d$, $\bar f_t$, } for all $t\geq 0$,
\begin{equation}\label{eq:sde_limit}
\left\lbrace \begin{array}{lll}
\ds \frac{d}{dt} \bar x_0(t) &\ds =& \ds - ( \bar \lambda_0 + \bar \mu_0) \bar x_0(t), \quad \bar x_0(0)=x^{in}_0\\
\ds \bar \lambda_0 &\ds =& \ds \sum_{x\in \N^d} \gl_0(x)\bar f_t(x),\\
\ds \bar \mu_0 &\ds =& \ds \sum_{x\in \N^d} \mu_0(x)\bar f_t(x),\\
\ds 0 &\ds =&\ds \sum_{x\in \N^d}\bar A_{\bar x_0(t)}\psi(x)\bar f_t(x), \quad \forall \psi \in B(\N^d).
\end{array}\right.
\end{equation}
\modF{In system \eqref{eq:sde_limit}, } $x^{in}_0$ is a real positive constant, and, for any $x_0>0$, $\bar A_{x_0}$ is an operator, defined, for all bounded \modF{functions} $\psi$ on $\N^d$ and for all $x=(x_1,\cdots,x_d)\in \N^d$, by

\begin{multline}
\label{eq:generator}
\bar A_{x_0}\psi(x) =\lambda_0(x)x_0 \left[\psi(x+e_1)-\psi(x)\right]+\sum_{i=1}^d \lambda_i(x)x_i \left[\psi(x-e_i+e_{i+1})-\psi(x)\right]\\
+\sum_{i=1}^d \mu_i(x)x_i \left[\psi(x-e_i)-\psi(x)\right]\,,
\end{multline} 
\modR{where, for $i \in \{1,\dots,d\}$, $e_i$ is a unit vector of $\N^d$, with coordinate $1$ in the $i^{th}$ position and zero elsewhere, and $e_{d+1}$ is the null vector.}

\begin{rem}
	Although the limit system \eqref{eq:sde_limit} may appear quite different from its deterministic counterpart \eqref{eq:ode_limit}, it has the same flavour: the fast variable is in a ``quasi-equilibrium'' at any time $t$. Its law $f_t$ thus needs to solve the equilibrium of the Kolmogorov equations associated \modF{with} SDE \eqref{eq:sde_rescaled}, which are written here with the help of the infinitesimal generator of the fast variable to ease the notations. More precisely, let $A^{\gge}$ \modF{be} the infinitesimal generator associated \modF{with} the process $X^{\gge}$, solution of \eqref{eq:sde_rescaled}, then for all \modF{functions} $\phi : \mathbb{R}_+ \times \mathbb{N}^d \to \mathbb{R} $ bounded and independent of the first variable ( $\forall (x_0,x) \in \mathbb{R}_+ \times \mathbb{N}^d, \quad \phi(x_0,x)=\psi(x)$), we have $$\forall (x_0,x) \in \mathbb{R}_+ \times \mathbb{N}^d, \quad A^{\gge} \phi(x_0,x)=\frac{1}{\gge} \bar A_{x_0}\psi(x).$$ 
	Thus $\overline f_t$ is the stationary solution \modF{corresponding} to $A^{\gge}$ when the slow variable is ``frozen''.
\end{rem} 

System \eqref{eq:sde_limit} is not necessarily well-posed, as there may be several solutions $\bar f_t$ for a given $\bar x_0$.  \modF{In the next two specific examples, we can prove that system (\ref{eq:sde_limit}) does admit a single solution, which is a natural} \modF{limit candidate} for the sequence $X^\gge$.

\begin{example}[Linear case]\label{ex:sde_limit_linear}
	
	Let us suppose that $K_{1,i}=K_{2,i}=0$ for all $i\in \{0,\dots,d\}$, \modR{and $(f_i+g_i)>0$ for all $i\in \{1,\dots,d\}$}. Then, system \eqref{eq:sde_limit} becomes linear and has a unique solution. 
	The invariant measure $\bar f_t$ has a product measure form 
	\begin{equation}\label{eq:_productform}
	\forall x \in \mathbb{N}^d, \quad \bar f_t (x_1,\dots, x_d) = \displaystyle \prod_{i=1}^{d} \bar f^i_t(x_i)\,,
	\end{equation}
	with $\bar f^i_t$ a Poisson law on $\mathbb{N}$ of mean parameter $p_ix_0(t)$, with
	\begin{equation}\label{eq:pi}
	p_i = \prod_{j=0}^{i-1} \frac{f_j}{f_{j+1}+g_{j+1}}\,.
	\end{equation} 
	System \eqref{eq:sde_limit} \modF{then reduces} to
	\begin{equation}\label{eq:sde_limit_linear}
	\left\lbrace \begin{array}{llll}
	\ds &\bar {x}_0(t) &\ds=&\ds x_0^{in}\exp\left(-(f_0+g_0)t\right)\\
	\ds \forall x \in \mathbb{N}^d, \quad &\bar  f_t(x) &\ds=& \ds \prod_{i=1}^{d} \, \left(p_i \, \bar x_0(t) \right)^{x_i} \,\frac{e^{-p_i\bar {x}_0(t)}}{x_i !}
	\end{array}\right.
	\end{equation}
	It is classical that stationary distributions associated \modF{with} the generator \eqref{eq:generator} are of product form for \modR{constant coefficients} $\lambda_i,\mu_i$ (see for instance \cite{Gadgil2005,Kingman1969,Kelly1979}). \modR{Taking the product form in Eq.~\eqref{eq:_productform} for granted, the following calculus shows that each marginal distribution has to be a Poisson law.}
	
	\modR{Indeed}, for a function $\psi$ which depends on the first variable only $$\forall x_1,\cdots, x_d, \quad \psi(x_1,\dots,x_d)=\psi_1(x_1)$$ and for $\bar x_0>0$, we obtain, using expression \eqref{eq:generator}, \begin{align*}\sum_{x\in \N^d} \bar A_{\bar x_0}\psi(x)\bar f_t(x) &= \sum_{x\in \N^d} \{\, [f_0 \bar x_0 (\psi_1(x_1+1)-\psi_1(x_1)) \,+\, (f_1 + g_1) x_1(\psi_1(x_1-1)-\psi_1(x_1)) \,] \\  & \quad \quad \quad \quad \bar f_t^1(x_1) \prod_{i=2}^{d}\bar f_t^i(x_i)\, \}  \\
	&= \big( \sum_{x\in \N} \{\,f_0 \bar x_0 (\psi_1(x+1)-\psi_1(x)) \,+\, (f_1 + g_1) x(\psi_1(x-1)-\psi_1(x)) \} \bar f_t^1(x) \big) \\  & \quad \quad \quad \quad \big(\sum_{(x_2,\dots,x_d) \in \N^{d-1}} \prod_{i=2}^{d}\bar f_t^i(x_i)\, \big) 
	\end{align*}
	Hence, the solution $\bar f_t$ is such that for any bounded function $\psi_1$
	\begin{align*}
	0&= \sum_{x\in \N} \psi_1(x)\{\,f_0 \bar x_0 \big( \bar f_t^1(x-1) \textbf{1}_{x\geq1} -\bar f_t^1(x) \big) \,+\, (f_1 + g_1) \big( \, (x+1) \,\bar f_t^1(x+1)- x \bar f_t^1 (x) \big) \} 
	\end{align*}
	\modF{This holds} in particular for $\psi_1(x)=\textbf{1}_{x=n}$ for any $n \in \mathbb{N}$\modF{, so that} we obtain after calculus $$\forall x \in \mathbb{N}, \quad \bar f_t^1(x) = \frac{1}{x!} \big(\frac{f_0 \, \bar x_0}{f_1 +g_1}\big)^x \bar f_t^1(0),$$
	and then, by same arguments
	$$\forall i \in \{1,\cdots,d\}, \forall x \in \mathbb{N}, \quad \bar f_t^i(x) = \frac{1}{x!} \big( p_i\, \bar x_0\big)^x \bar f_t^i(0)$$ where $p_i$ \modF{is defined} in Eq.~\eqref{eq:pi}.
\end{example}

\begin{example}[A single feedback \modF{onto the quiescent follicle} death rate] 
	Let us suppose that $K_{1,i}=0$ for all $i\in \{0,\dots,d\}$ and $K_{2,i}=0$ for $i\in \{1,\dots,d\}$ but $K_{2,0}>0$. Then, system \eqref{eq:sde_limit} \modR{can be simplified as} 
	\begin{equation}\label{sde_ex2}
	\left\lbrace \begin{array}{llll}
	\ds \frac{d}{dt} \bar x_0(t) &\ds =& \ds - ( f_0 + \bar \mu_0(\bar x_0(t))) \bar x_0(t), \\
	\ds \bar \mu_0(\bar x_0) &\ds =& \ds g_0(1 + \bar x_0 K_{2,0} \sum_{j=1}^d b_j p_j),\\
	\ds \forall x \in \mathbb{N}^d, \quad &\bar  f_t(x) &\ds= \ds \prod_{i=1}^{d} \, \left(p_i \, \bar x_0(t) \right)^{x_i} \,\frac{e^{-p_i\bar {x}_0(t)}}{x_i !},
	\end{array}\right.
	\end{equation}
	where $p_i$ is defined in Eq.~\eqref{eq:pi}, and has a unique solution. \modR{The justification of system \eqref{sde_ex2} follows that of Example \eqref{ex:sde_limit_linear} for the measure $f_t$, which is not directly modified by the feedback term onto the quiescent follicle death rates.}
\end{example}

\subsection{Convergence in the linear case}

In this paragraph, we assume that $K_{1,i}=K_{2,i}=0$ for all $i\in \{0,\dots,d\}$ as in Example \ref{ex:sde_limit_linear}, and we assume \modF{further} that the initial condition $X^{\varepsilon,in}=\left(X_0^{\varepsilon,in},X_1^{\varepsilon,in},\cdots,,X_d^{\varepsilon,in}\right)$ of system \eqref{eq:sde_rescaled} is such that $\frac{1}{\gge}X_0^{\varepsilon,in}$, $X_1^{\varepsilon,in}$,..., $X_d^{\varepsilon,in}$ are independent Poisson random \modF{variables} of mean respectively $x_0^{in}/\gge$, $x_1^{in}$,..., $x_d^{in}$. Denoting by $G^{in}$ the probability generating function of the integer-valued random vector $\left(\frac{1}{\gge}X_0^{\varepsilon,in},X_1^{\varepsilon,in},\cdots,,X_d^{\varepsilon,in}\right)$, we thus have
\begin{equation}\label{eq:IC_stoch_linear}
G^{in}(z)=\mathbb{E}\left\{z_0^{X_0^{\varepsilon,in}/\gge}\prod_{j=1}^d z_j^{X_j^{\varepsilon,in}}\right\}=\exp\left(\frac{x_0^{in}}{\gge}(z_0-1)\right)\prod_{i=1}^d \exp\left(x_i^{in}(z_i-1)\right)
\end{equation}
In such \modF{a} case, one can solve explicitly system \eqref{eq:sde_rescaled} for each $\gge>0$. \modR{We briefly sketch the formal arguments and computations (see \cite{Gadgil2005} for details)}. We define the probability generating function of the vector $\left(X_0^{\varepsilon}(t)/\gge,X_1^{\varepsilon}(t),\cdots,,X_d^{\varepsilon}(t)\right)$, for $\textbf{z}\in \mathbb{R}^{d+1}$, by
\begin{equation}\label{eq:mgf}
G^\gge(\textbf{z},t)=\mathbb{E}\left\{z_0^{X_0^\gge(t)/\gge}\prod_{j=1}^d z_j^{X_j^{\varepsilon}(t)}\right\}=\sum_{i=0}^d\sum_{n_i=0}^\infty \prod_{j=0}^d z_j^{n_j} \mathbb{P}\left\{X_0^\gge(t) =\gge n_0,X_1^\gge(t) =n_1,\cdots,X_d^\gge(t) =n_d \right\}
\end{equation}
The infinitesimal generator $\tilde A^\gge$ of the process $\tilde X^\gge(t):=\left(X_0^{\varepsilon}(t)/\gge,X_1^{\varepsilon}(t),\cdots,,X_d^{\varepsilon}(t)\right)$, is given\modF{, for all bounded functions} $\psi$ on $\N^{d+1}$ and for all $n=(n_0,n_1,\cdots,n_d)\in \N^{d+1}$, by
 \begin{multline}
 \label{eq:generator_2}
 \tilde A^\gge\psi(n) =f_0 n_0 \left[\psi(n-e_0+e_1)-\psi(n)\right]+g_0n_0 \left[\psi(n-e_0)-\psi(n)\right]\\
+ \sum_{i=1}^{d-1} \frac{f_i}{\gge}n_i \left[\psi(x-e_i+e_{i+1})-\psi(x)\right] +\sum_{i=1}^d \frac{g_i}{\gge}n_i \left[\psi(x-e_i)-\psi(x)\right]\,.
 \end{multline} 
Using the Kolmogorov backward equation,
\begin{equation}
\frac{d}{dt}\mathbb{E}\left[\psi\left(\tilde X^\gge(t)\right)\right]=\mathbb{E}\left[\tilde A^\gge\psi\left(\tilde X^\gge(t)\right)\right]
\end{equation}
with $\psi(n)=\prod_{j=0}^d z_j^{n_j}$  \modR{(a truncation procedure is required to deal rigorously with such test functions)}, we obtain, using linearity of expectation, a first-order partial differential equation on $G^\gge$ , given by 
\begin{multline}\label{eq:pde_mgf}
\frac{\partial}{\partial t}G^\gge(\textbf{z},t) = -(f_{0}+g_0)(z_0-1) \frac{\partial}{\partial z_{0}}G^\gge(\textbf{z},t)
+f_0(z_1-1)\frac{\partial}{\partial z_{0}}G^\gge(\textbf{z},t)\\
+\sum_{i=1}^{d-1} (z_{i+1}-1)\frac{f_i}{\gge} \frac{\partial}{\partial z_{i}}G^\gge(\textbf{z},t)-\sum_{i=1}^d (z_i-1)\frac{1}{\gge}(f_i+g_i) \frac{\partial}{\partial z_{i}}G^\gge(\textbf{z},t)
\end{multline} 
It turns out that the unique solution \modF{to} \eqref{eq:pde_mgf} is given by 
\begin{equation}\label{eq:solution_poisson}
G^\gge(z,t)= \exp\left(\frac{x_0^\gge(t)}{\gge}(z_0-1)\right)\prod_{i=1}^d \exp\left(x_i^\gge(t))(z_i-1)\right)\,.
\end{equation}
where $(x_i^\gge(t))_{i=0,\cdots,d}$ is solution of the very same ODE as the linear version of Eq.~\eqref{eq:ode_rescaled} (whose \modF{solutions} are given by Eq.~\eqref{eq:ode_rescaled_solution_linear}), with initial condition $x_i^\gge(t=0)=x_i^{in}$, for all $i=0...d$. Thus, at any time $t$, $X_0^{\varepsilon}(t)/\gge$, $X_1^{\varepsilon}(t)$,..., $X_d^{\varepsilon}(t)$ are independent Poisson random variables of mean (respectively) $\frac{x_0^\gge(t)}{\gge}$, $x_1^\gge(t)$, ..., $x_d^\gge(t)$.
	
	\begin{prop}\label{prop:sde_cv_linear} Assume that $K_{1,i}=K_{2,i}=0$ for all $i\in \{0,\dots,d\}$, that $f_i+g_i$ is strictly positive for $i\in \{1,\dots,d\}$ and that the initial condition $X^{\varepsilon,in}$ is such that Eq.~\eqref{eq:IC_stoch_linear} holds. Then, for all $t>0$, $X_0^{\varepsilon}(t)$ converges (in law) towards the deterministic value $\bar x_0(t)=x_0^{in}\exp(-(f_0+g_0)t)$, and, for all $i=1,\cdots,d$, $X_i^{\varepsilon}(t)$ converges (in law) towards a Poisson random variable of mean $\bar x_i(t)= p_i \bar x_0(t)$, where $p_i$ is given by Eq.~\eqref{eq:pi}.
	\end{prop}
	
	\begin{proof}
	The proof is a direct consequence of the explicit solution \eqref{eq:solution_poisson} for the probability generating function $G^\gge$ \modF{combined with} Proposition \ref{prop:ode_cv_linear}. As $X_0^{\varepsilon}(t)/\gge$ is a Poisson random variable of mean $\frac{\bar x_0(t)}{\gge}$, it is clear that $\mathbb{E}\left[X_0^{\varepsilon}(t)\right]=\bar x_0(t)$ and $var\left(X_0^{\varepsilon}(t)\right)=\gge \bar x_0^2(t)\leq \gge x_0^{in}\to0$ as $\gge\to0$, which implies that $X_0^{\varepsilon}(t)$ converges in law towards $\bar x_0(t)$. For all $i=1,\cdots,d$, and $t>0$, as $x_i^{\varepsilon}(t)$ converges to $p_i \bar x_0(t)$ as $\gge\to\modF{0}$, $X_i^{\varepsilon}(t)$ converges (in law) towards a Poisson random variable of mean $p_i \bar x_0(t)$.
	\end{proof}	
	
	\begin{rem}\label{rem:sde}

The slow variable $X_0^{\gge}$ is independent of the fast variables and its \modF{dynamics are} reduced to a ``death and death'' process with constant rate. Then convergence of this variable is given by Theorem 8.1 of \cite{kurtz}, which \modF{enables us to obtain stronger results with fewer hypotheses} on the initial condition. Suppose that 
$$\displaystyle \lim_{\gge \to 0} X_0^{\gge,in}=x_0^{in} \quad a.s$$ 
then 
$$\lim_{\gge \to 0} \sup_{s \leq t} \vert X_0^{\gge}(s) - \bar x_0(s)\vert = 0 \quad a.s \; \; \text{ for all } t >0.$$

	\end{rem}
	
    \subsection{\modFF{Numerical convergence}}
    
    In this paragraph, we illustrate the convergence of $(X_0^\gge, X_1^\gge,\cdots, X_d^\gge)$ as $\varepsilon\to0$.  The chosen scenario and the parameter values are detailed in \modF{the Appendix} (section \ref{sec:annex}).
    
    		\begin{figure}
    			\includegraphics{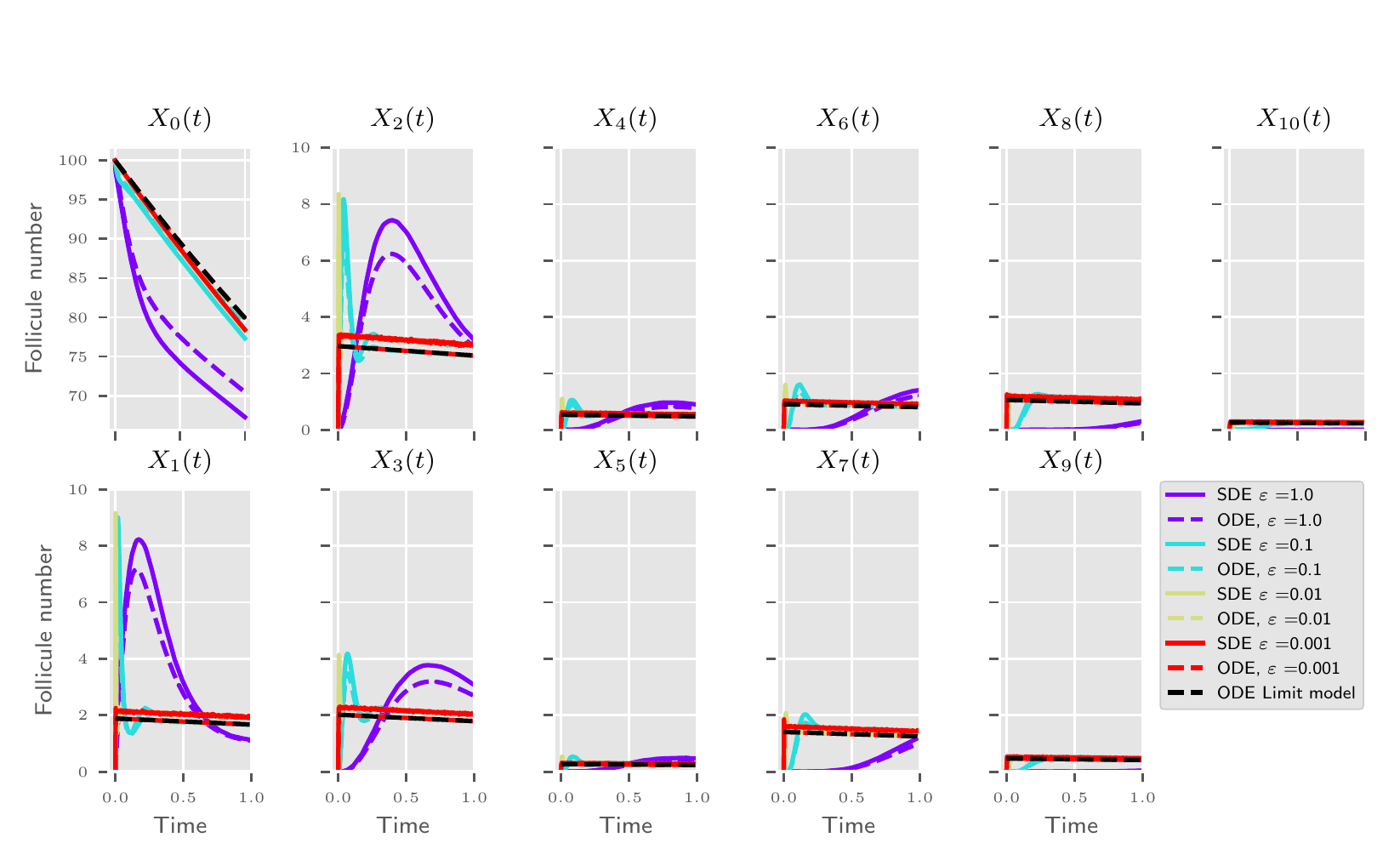}
    			\caption{\modF{In the same way as} in Figure \ref{fig:ode_lim_mod}, we plot the trajectories \modF{in each maturity} compartment ($d=10$), for the rescaled variables $X_i^\varepsilon$ of the SDE system \eqref{eq:sde_rescaled} (\modF{solid} lines) and the ODE system \eqref{eq:ode_rescaled} (dashed lines), for different $\varepsilon$ (see legend). For the SDE, we plot the empirical mean computed over 10000 trajectories. \modF{The} limit variable $\bar X_i$  of the ODE system \modF{corresponds to} the black dashed line.}
    			\label{fig:sde_lim_mod}
    		\end{figure}
    		
    		\begin{figure}
    			\includegraphics{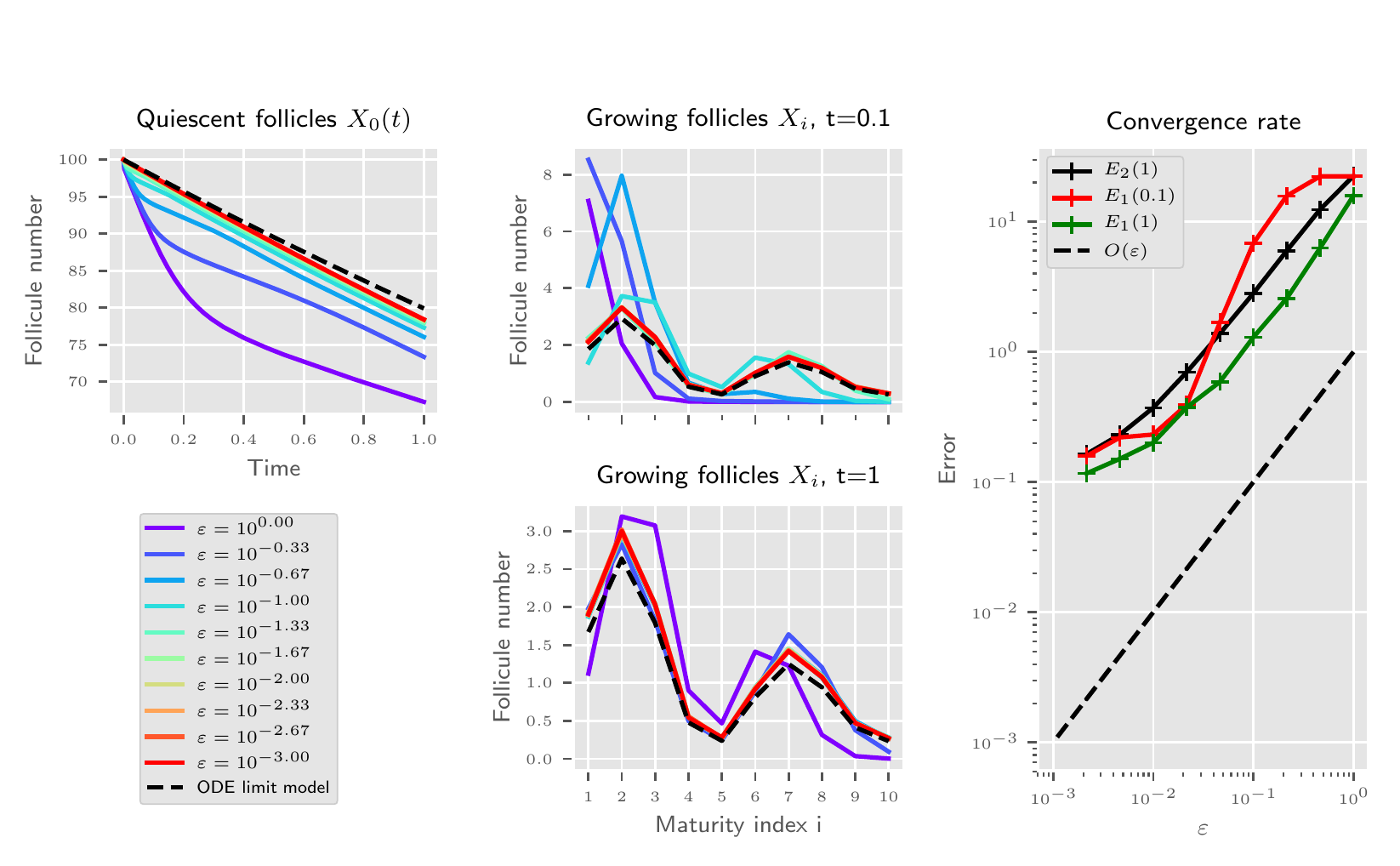}
    			\caption{Trajectories \modF{in} the quiescent follicle compartment (top left panel) and \modF{distribution} of the growing follicle population at time $t=0.1$ (top center panel) and at time $t=1$ (bottom center panel), for the rescaled variables $X_i^\varepsilon$, for different $\varepsilon$ (\modF{solid} colored lines, see \modF{legend insert}) and the limit variable $\bar X_i$ (\modF{black dashed line}). On the right panel, we plot the discrete $l^1$ norm error $E_1(t)$ at a \modF{fixed time} $t=0.1$ and $t=1$ (\modF{solid} red and green lines, resp.) and the $l^1$-cumulative error $E_2(1)$ on $t\in(0,1)$ (black \modF{solid} line) as a function of $\varepsilon$ (see details in \modF{body} text). The black dashed line is the \modF{straight line of slope 1 according} to $\varepsilon$.}
    			\label{fig:sde_lim_mod_hist}
    		\end{figure}
    
    In Figure \ref{fig:sde_lim_mod}, we plot the empirical mean trajectories \modR{(computed over $10^4$ sampled trajectories)} \modF{in each maturity} compartment ($d=10$) for the rescaled model on a time horizon $t\in(0,1)$ in the nonlinear scenario ($K_{1\modF{,}0}>0$), together with the trajectories of the analogous ODE rescaled system \eqref{eq:ode_rescaled} and its limit \eqref{eq:ode_limit}. \modR{We observe that, for each compartment, the empirical mean of the SDE seems to converge \modFF{to a limit value}, yet \modFF{this limit} does not superimpose with the ODE limit solution (which is expected in a nonlinear scenario, as the ODE and SDE limits are different).}
    
    In Figure \ref{fig:sde_lim_mod_hist}, \modF{using} the same parameters as in Figure \ref{fig:sde_lim_mod}, we \modF{display the maturity distribution in} the growing follicle population, for various $\varepsilon$. We empirically quantify the convergence rate using the following error, at time $t$,
    \begin{equation}
    E_1(t)= \sum_{i=0}^d \mid \mathbb{E} X_i^{\varepsilon}(t)-\mathbb{E}\bar X_i(t) \mid\,,
    \end{equation}
    and the cumulative error on time interval $[0,T]$,
    \begin{equation}
    E_2(T)= \int_0^T \sum_{i=0}^d \mid \mathbb{E}X_i^{\varepsilon}(t)-\mathbb{E}\bar X_i(t) \mid dt\,,
    \end{equation}
    \modF{which can be assessed numerically as}
    \begin{equation}
    \tilde E_2(T)= \sum_{k=0}^{N_T} \delta_t \sum_{i=0}^d \mid \mathbb{E}X_i^{\varepsilon}(t_k)-\mathbb{E} \bar X_i(t_k) \mid\,,
    \end{equation}
    where $t_k=k\delta_t$, for $k=0\modF{\cdots}N_T$. \modFF{In practice, we} also replace the limit model $\bar X$ by the numerically evaluated \modF{limit model} $ X^\varepsilon$ with $\varepsilon=0.001$. We then observe that the error \modF{decreases roughly} linearly with $\varepsilon$.
    
    \begin{figure}[h!]
    			\includegraphics{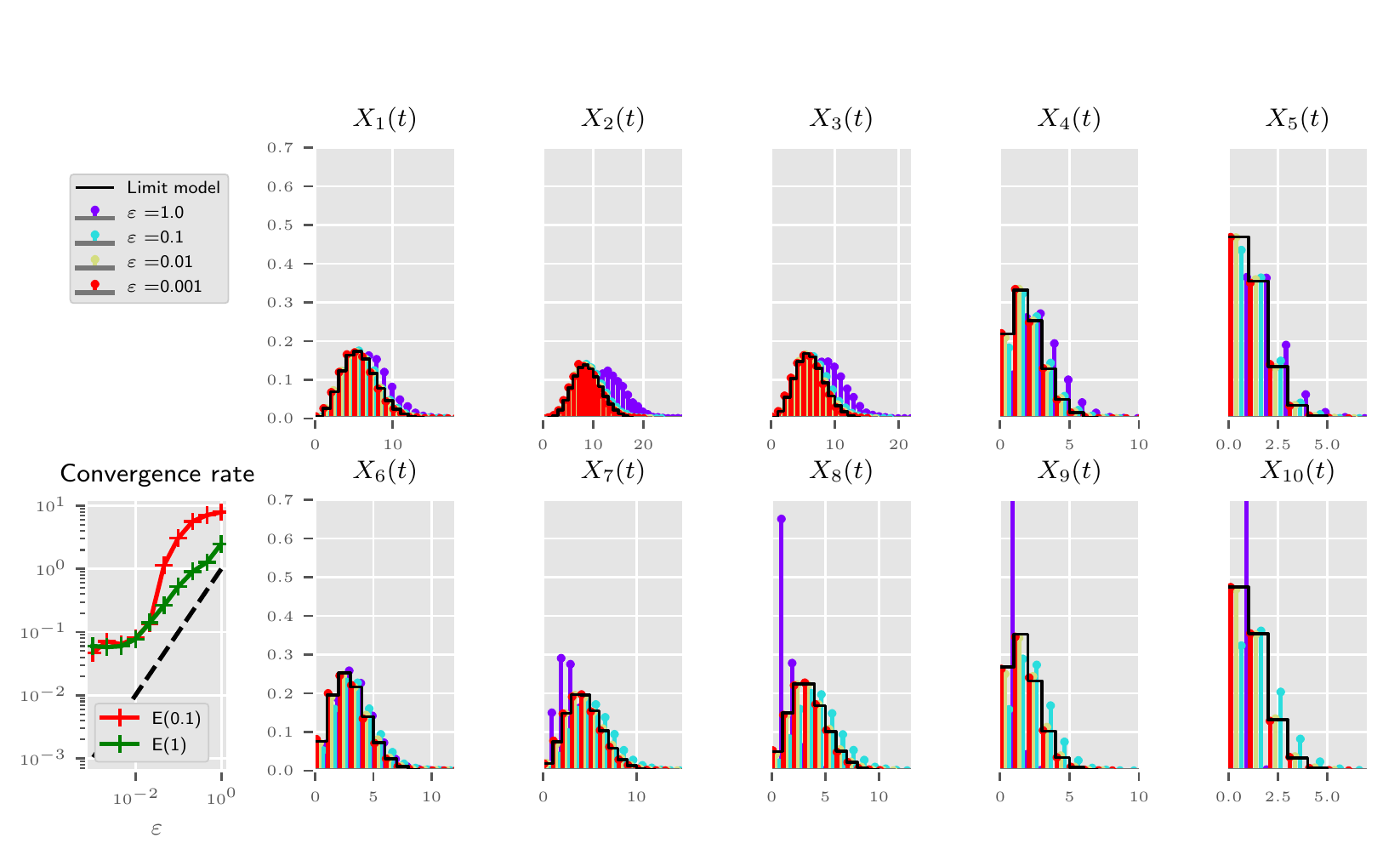}
    			\caption{Empirical law \modF{of $X_i^\varepsilon$ in each maturity compartment} at time $t=1$ for different $\varepsilon$ in colored bars (see legend \modF{insert}) and the limit distribution $\bar X_i$ (\modF{black solid} lines). On the bottom leftmost panel, we plot the total variation error $E(t)$ at a \modF{fixed time} $t=0.1$ and $t=1$ (\modF{solid} red and green lines, resp.). The black dashed line is the \modF{straight line of slope 1 according} to $\varepsilon$.}
    			\label{fig:sde_lim_mod_law}
    \end{figure}
    		
    In Figure \ref{fig:sde_lim_mod_law}, we use the linear scenario ($K_{1\modF{,}0}=0$) detailed in \modF{the Appendix} (section \ref{sec:annex}) to visualize the convergence of the fast variable of the rescaled SDE to the ``quasi-stationary'' distribution of the limit model. The marginals of the rescaled model are evaluated \modR{over $10^4$ sampled trajectories at time $t=0.1$ and $t=1$.}
    The \modF{errors} between the marginal laws of the rescaled and limit models \modF{are} quantified by the total variation (restricted on the support of the \modF{numerically assessed limit} model):
    \begin{equation}
    E(t)= \sum_{i=1}^d d_{TV}(X_i^{\varepsilon}(t),\bar X_i(t))\,,
    \end{equation}
    where
    \begin{equation}
    d_{TV}(X,Y)=\max\{\mid \pi_X(i)-\pi_Y(i) \mid\,, i\in \mathbb{N}\}\,,
    \end{equation}
    and $\pi_X$, $\pi_y$ are, respectively, the \modF{laws} of $X$ and $Y$. \modR{The error seems to decrease in a sub-linear manner with $\varepsilon$, with a plateau for $\gge<10^{-2}$, which is probably due to the limited finite sampling size ($10^4$). }
    
\section{\modF{Appendix} - parameter values}\label{sec:annex}
\modF{In the numerical illustrations provided throughout the previous sections, we refer to either a linear or nonlinear scenario. As far as parameter values, the only difference is that parameter $K_{1,0}$ is set to $0$ in the linear scenario. All other parameters are identical and chosen as explained below.}
\subsection*{\modF{Numerical simulation of} the PDE model}

\modF{We begin by shaping the desired solution $H(x)$ and we choose functions $f$ and $g$ accordingly:}
	$$H(x)=\displaystyle e^{-\displaystyle\int^x_0 \displaystyle\frac{g(y)+f'(y)}{f(y)}dy}.$$
 $$\modF{\mbox{Choosing }}\,g(x)=1,\quad\mbox{\modF{we get} }\quad f(x)=\displaystyle \frac{c-\int_0^x H(y)dy}{H(x)}$$

\modF{Motivated by our biological application, we more specifically select a two-bump function}
	$$H(x)=\dfrac{p_1e^{-\dfrac{(x-x_1)^2}{2s^2}} +p_2e^{-\dfrac{(x-x_2)^2}{2s^2}}}{p_1e^{-\dfrac{x_1^2}{2s^2}} +p_2e^{-\dfrac{x_2^2}{2s^2}}}$$
	with $s=0.1$, $x_1=0.2$, $p_1=0.7$, $x_2=0.7$ and $p_2=0.3$. 
	
	\modF{Except $K_{1,0}$, all coefficients weighting the nonlinear terms,  $K_{20}$, $K_1$, $K_2$, are set to zero. As a result, functions $b$, $\omega_1$ and $\omega_2$, that were introduced for the sake of genericity, are not used in the numerical illustrations.}
	
	The basal \modF{activation rate $f_0$ is set to $f_0=1$}, and the basal death rate \modF{in quiescent follicles $g_0$ is set to} $g_0=0.1$.\\ 
	\modF{Since the population feedback onto the activation rate is mainly exerted by follicles in an intermediate maturity stage}, we choose $a(x)=1_{[0.3,0.7]}$, \modF{given that the state space lies in $x\in[0,1]$}. The \modF{feedback gain} is \modF{set to} $K_{1,0}=2$.

\modF{Finally, the time horizon covers $t\in(0,1)$, and the initial condition is given by} $\rho_0^{ini} = 100$ and $\rho^{ini} \equiv 0$.

 The parameter values are \modF{summed} up in Table \ref{tab_param_1} \modF{and illustrated on Figure \ref{fig:param}}.
 
	\begin{table}[H]
		\begin{tabular}{|c|c|c|c|c|c|c|c|}
			\hline
			$s$  & $p_1$  & $x_1$ & $p_2$ & $x_2$& $c$& & \\
			\hline
			0.1 & 0.7  &0.2 & 0.3 & 0.7 &2.61& & \\
			\hline
		$f_0$  & $g_0$  & $K_{1\modF{,}0}$ & $K_{2\modF{,}0}$ & $\rho^0_{ini}$&$a$  & $K_{1,d},d\neq 0$ & $K_{2,d},d\neq 0$\\
			\hline
			1 & 0.1  &2 & 0 & 100& $1_{[0.3,0.7]}$ & 0 & 0\\
			\hline
		\end{tabular}
		\caption{Parameter values for the numerical simulations}
		\label{tab_param_1}
	\end{table}
	
\begin{figure}[h!]
	\includegraphics{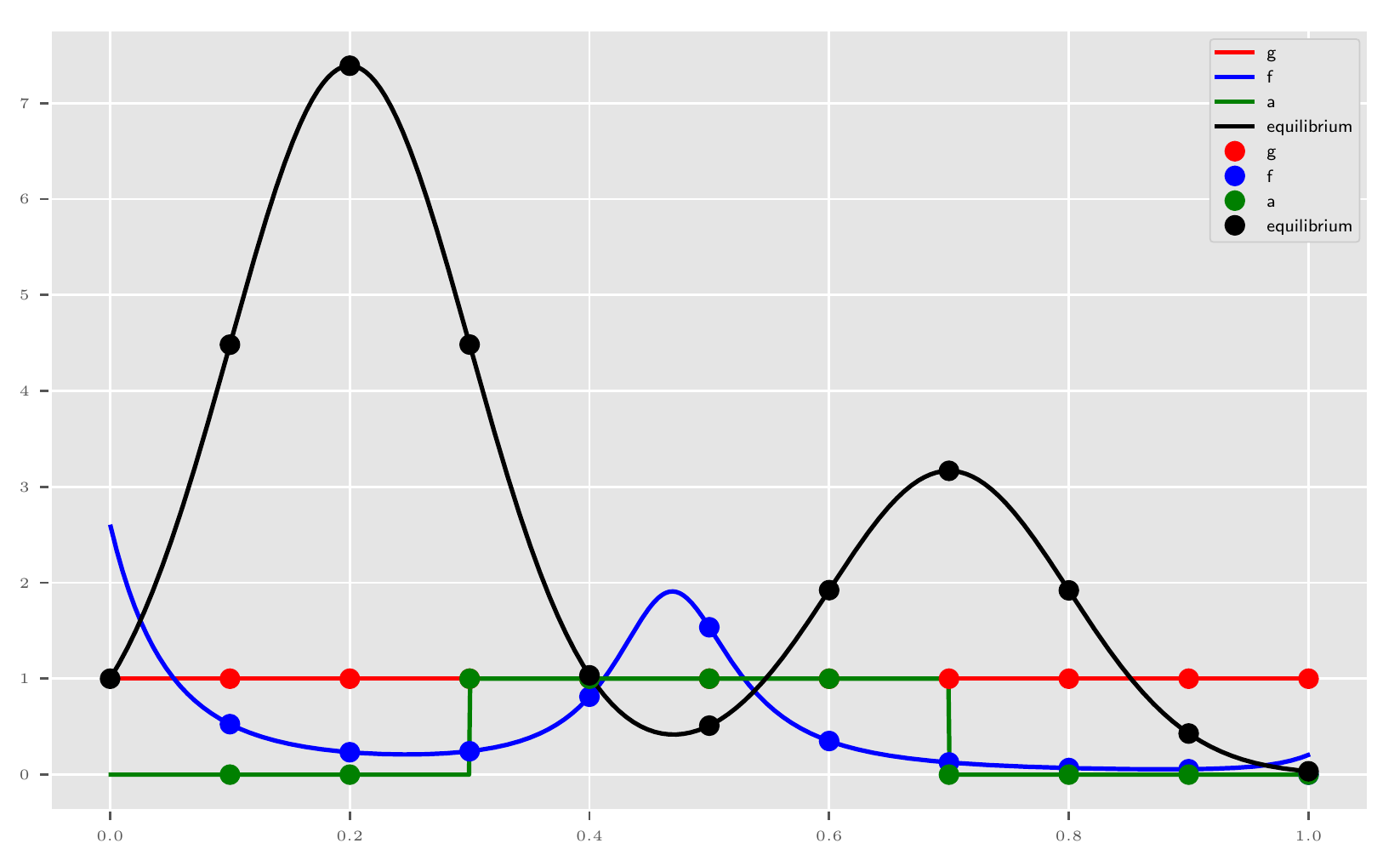}
	\caption{\modF{Parametric functions} used in the numerical simulations, for the PDE (plain lines), the ODE and the SDE (dot symbols).}
	\label{fig:param}
\end{figure}

\subsection*{\modF{Numerical simulation of the ODE/SDE} model}

For the \modF{ODE- and SDE-based models, we first set the number of compartments}, $d=10$, and define $x_i = i/d$. Then, we \modF{adapt the functions selected in the continuous PDE case and set}, for $i\in \{1,d\}$, 
\begin{itemize}
\item $a_i=a(x_i)$
\item $f_i=f(x_i)/d$
\item $g_i=g(x_i)$
\item $K_{1,i}=K_{2,i}= 0$
\end{itemize}
All other parameters \modF{are kept as in }Table \ref{tab_param_1}, \modF{while the initial condition is chosen as $Y_0^{ini} = 100$ and $Y_i^{ini}= 0$, $i\in\{1,d\}$}.

\modF{To simulate the ODE model}, we use the standard python scipy.odeint, \modF{while, to simulate the SDE model}, we use an exact stochastic simulation algorithm (Gillespie).




\bibliographystyle{plain}
\bibliography{biblio_cemracs}
\end{document}